\documentclass[reqno]{amsart}
\usepackage{amsthm,amsmath,amsfonts,amssymb,amscd,mathrsfs,graphics}
\usepackage{txfonts}
\usepackage{hyperref,supertabular}

\DeclareMathOperator{\BCOV}{BCOV}
\DeclareMathOperator{\dR}{dR}
\newcommand{\C}{{\mathbb{C}}}

\newcommand{\D}{{\mathscr D}}

\newcommand{\E}{{\mathscr E}}
\newcommand{\End}{{\operatorname{End}}}

\newcommand{\Hc}{\mathcal H}

\newcommand{\IH}{{\mathbb{H}}}

\newcommand{\N}{{\mathbb N}}

\newcommand{\IP}{{\mathbb{P}}}
\newcommand{\Q}{{\mathbb{Q}}}

\newcommand{\R}{{\mathbb{R}}}

\newcommand{\Tr}{{\operatorname{Tr}}}

\newcommand{\Z}{{\mathbb{Z}}}
\newcommand{\alphab}{{\underline{\alpha}}}

\newcommand{\bu}{{\mathbf u}}

\newcommand{\diag}{\operatorname{diag}}

\newcommand{\e}{{\mathbf E}}
\newcommand{\fb}{{\bar{f}}}

\newcommand{\ib}{{\bar{i}}}

\newcommand{\im}{\operatorname{Im}}

\newcommand{\jb}{{\bar{j}}}

\newcommand{\kb}{{\bar{k}}}
\newcommand{\lb}{{\bar{l}}}

\newcommand{\mub}{{\bar{\mu}}}

\newcommand{\nub}{{\bar{\nu}}}

\newcommand{\om}{{\omega}}

\newcommand{\q}{{\mathbf q}} 

\newcommand{\re}{\operatorname{Re}}

\newcommand{\sgn}{\operatorname{sgn}}

\newcommand{\tr}{\operatorname{tr}}

\newcommand{\w}{{\mathbf w}}
\newcommand{\x}{{\mathbf x}}
\newcommand{\z}{{\mathbf z}}
\newcommand{\zb}{{\bar{z}}}

\newcommand{\pat}{{\partial}}
\newcommand{\bpat}{{\bar{\partial}}}
\newcommand{\hc}{\hat c}
\newcommand{\bc}{\bar c}
\newcommand{\hbc}{\hat{\bar c}}
\newcommand{\str}{\operatorname{str}}
\newcommand{\Str}{\operatorname{Str}}

\newtheorem{thm}{Theorem}[section]
\newtheorem{lm}[thm]{Lemma}
\newtheorem{prop}[thm]{Proposition}
\newtheorem{crl}[thm]{Corollary}

\theoremstyle{definition}
\newtheorem{rem}[thm]{Remark}

\newtheorem{df}[thm]{Definition}
\newtheorem{ex}[thm]{Example}
\theoremstyle{remark}

\begin{document}

\title{Torsion type invariants of singularities}
\date{\today}
\author{Huijun Fan$^\dag$}
\address{School of Mathematical Sciences, Peking University, Beijing, China}
\email{fanhj@math.pku.edu.cn}
\author{ Hao Fang$^*$}
\address{Department of Mathematics, University of Iowa, Iowa city, IA, USA}
\email{hao-fang@uiowa.edu}
\thanks{$^\dag$ Supported by NSFC (11271028, 11325101) and Doctoral Fund of Ministry of Education of China (20120001110060).}
\thanks{$^*$ Partially supported by NSF DMS-100829. }

\maketitle
\begin{abstract} Inspired by the LG/CY correspondence, we study the local index theory of the Schr\"odinger operator associated to a  singularity defined on $\C^n$ by a quasi-homogeneous polynomial $f$. Under some mild assumption on $f$, we show that the small time heat kernel expansion of the corresponding Schr\"odinger operator exists and is a series of fractional powers of time $t$. Then we prove a local index formula which expresses the Milnor number of $f$ by a Gaussian type integral. Furthermore, the heat kernel expansion provides  spectral invariants of $f$. Especially, we  define  torsion type invariants associated to a singularity.  These spectral invariants provide a new direction to study the singularity. 
\end{abstract}

\section{Introduction}

In this article, we study the local index theory of the Schr\"odinger operator associated to a  singularity defined on $\C^n$ by a quasi-homogeneous polynomial $f$. 

The pair $(\C^n,f)$ provides the input data of the supersymmetric Landau-Ginzburg (LG) model in physics. The topological field theory of the supersymmetric Landau-Ginzburg model has been studied by physicists in the end of 1980s. For example, in \cite{CGP1,CGP2, Ce1, Ce2}, the chiral ring structure was computed and  in \cite{CV} the $tt*$-structure was studied. It turns out that the LG model is closely related to the classical singularity theory and generates many interesting mathematical structures (see more detail in Section 2). 

 The LG model $(\C^n,f)$ is a special example of the so called section-bundle system introduced in \cite{Fa}. A section bundle system $(M,g, f)$ consists of a complete non-compact K\"ahler manifold $M$ with metric $g$, which has bounded geometry, and a holomorphic function $f$ defined on $M$. Motivated by the early work of Cecotti-Vafa \cite{CV}, the first-named author has studied the geometry of the twisted Cauchy-Riemann operator $\bpat_f=\bpat+\pat f\wedge$ on the strongly tame section bundle system~\cite{Fa}. He obtains the L$^{2}$ spectral theory for the corresponding Laplacian and established a Hodge theory. It is thus natural to discuss  torsion type invariants for this construction.

On a very different note, a particular type of analytic torsion plays a crucial role in the mirror symmetry for  Calabi-Yau (CY) manifolds. 

Torsion invariants have been introduced in topology to distinguish delicate topological spaces~\cite{Re, Fr, Mi}. Developed by Ray and Singer~\cite{RS1}, analytic torsion is an analog construction on chain complexes of Hilbert spaces which has profound impact in recent development of geometry (see Section 2 for more details and references). In particular, Bershadsky-Cecotti-Ooguri-Vafa have defined the so-called BCOV torsion on CY manifolds, which is of crucial importance in the $N=2$ superconformal field theory. In~\cite{FLY}, the second-named author and his collaborators, Lu and Yoshikawa, have computed the explicit anomaly formula for BCOV torsion; they have
examined  asymptotic behaviors of the BCOV torsion on the polarized moduli space of CY 3-folds near singular CY varieties.  Guided by the recent development in the LG/CY correspondence in the global mirror symmetry picture (see \cite{CR}), it is desirable for us to define and study the LG analogue of the BCOV torsion.
 
It is thus our goal to develop the heat kernel analysis and local index theory for the $\bpat_{f}$ operator and its corresponding chain complex. Let us state some of our main results.
 
Let $\Delta_f$ be the corresponding Laplacian. Firstly, under very mild constraints to a strongly tame system $(M,g,f)$, we prove that $e^{-t\Delta_f}$ is of trace class.  

Secondly, we consider $(\C^n, f)$ with a non-degenerate quasi-homogeneous polynomial $f$. Let $q_i$ be the weights of $f$ with respect to $z_i,i=1,\cdots,n$. Denote by $q_M=\max_i\{q_i\}, q_m=\min_i\{q_i\}$, and $\delta=\frac{1-3(q_M-q_m)}{3(1-q_M)}$. Let 
$$
P^0_k(\z,\w,t):=P_k(\z,\w,t)=\E_0(\z,\w,t)\E_1(\z,\w,t)\sum^k_{i=0}t^i U_i(\z,\w)
$$
be the parametrix of $L:=\pat_t+\Delta_f$ as defined in (\ref{sect-4-equa-0}) which satisfies $LP_k=R_k$. Let $P^i_k:=P_k^{i-1}*R_k$ as defined in (\ref{sect-4-equa-61}). Then we have the following:

\begin{thm}[Theorem \ref{sect-3-main-theorem-1}]\label{sect-1-theo-1} Let $f$ be a non-degenerate quasi-homogeneous polynomial on $\C^n$ satisfying $q_M-q_m<\frac{1}{3}$ and let $k\in \N$ satisfies $k>\frac{3l_0+n+1+2\sum_i q_i}{\delta}$. Fix $T>0$. Then for any $t\in (0,T]$, the series
$$
P(\z,\w,t)=\sum^\infty_{i=0}(-1)^i P^i_k(\z,\w,t)
$$
converges for any $(\z,\w)\in \C^n\times \C^n$, and further $P(\z,\w,t)$ has up to $2l_0$-order $\z$-derivatives and up to $l_0$ order $t$ derivatives. $P(\z,\w,t)$ is the unique heat kernel of the operator $\pat_t+\Delta_f$.  
\end{thm}

Note that the existence of the heat kernel expansion $\tr(e^{-t\Delta_f})$ as $t\to 0$ is not obvious due to the non-compactness of the manifold $M$.   In our case, $\Delta_f$ is a Schr\"odinger operator of matrix type having the following form:
$$
\Delta_f=\Delta_{\bpat}+L_f+|\pat f|^2.
$$
Hence the heat kernel expansion depends not only on the local information, but also depends on the growth of the potential function $|\pat f|^2$ at infinity. There are many results on the heat kernel expansion of scalar Schr\"odinger operator with upper bounded potential function (ref. \cite{HP,GW}), but there are seldom result about the heat kernel expansion of the matrix Schr\"odinger operator with unbounded potential function.   

\begin{rem}The singularities satisfying the condition $q_M-q_m<\frac{1}{3}$ include all homogeneous polynomials, ADE simple singularities, and the unimodal singularities which have three types: parabolic type, hyperbolic type and 14 Exceptional families (see \cite[Vol. I, page 246]{AGV}). However, it is easy to find singularities which do not satisfying this condition, e.g., $z_1^2+z_1 z_2^3+z_2 z_3^3,z_1^2+z_1z_2^2+z_2 z_3^4$ which have weights $(\frac{1}{2},\frac{1}{6},\frac{5}{18})$ and $(\frac{1}{2},\frac{1}{4},\frac{3}{16})$ respectively.
\end{rem}

Applying the heat kernel expansion, we  obtain the following index formula regarding the L$^{2}$ $\bpat_f$-cohomology. Notice that $(-1)^{n}\mu(f)$ is the Euler characteristic of the above mentioned cohomology. We get the following conclusion.

\begin{thm}[Corollary \ref{sect-5-coro-index1}] \label{add2}Under the condition of Theorem \ref{sect-1-theo-1}, the Milnor number $\mu(f)$ is given by the Gaussian type integral
\begin{equation}
\mu(f)=\frac{t^n}{\pi^n}\int_{\C^{n}}\exp (-t|\pat f|^{2})|\det \partial^{2}f|^{2}d\text{vol},\forall t>0,
\end{equation}
if we take $f(z)=z^2/2$, then we obtain the well-known formula:
\begin{equation}
\pi^n=\int_{\C^n}e^{-|\z|^2}d\text{vol}.
\end{equation}
\end{thm}
It is thus natural for us to consider torsion type invariants. Define the $i$-th zeta function $\Theta^i_f(s)$

$$
\Theta^i_{f}(s)=\frac{1}{2}\sum_{k=0}^{2n}(-1)^k k^{i}\frac{1}{\Gamma(s)}\int^\infty_0 \Tr(e^{-t\Delta^k_{f}}-\Pi)t^{s-1}dt,
$$ which is analytical in the domain $\{s\in \C|\re(s)>2n\}, n=\dim_\C M$.  Due to our trace class result and the asymptotic formula for heat kernel, we define the $i$-th torsion type invariant
\begin{equation}
\log T^i(f)=-(\Theta^{R,i}_f)'(0).
\end{equation}
Thus, $T^{1}_{f}$ is the natural torsion for the chain complex associated to the $\bpat_{f}$, which is shown to be constant by Theorem \ref{add1}. It is however interesting to point out that $T^{2}_{f}$ is an analogue of the BCOV torsion in this setting. 

\begin{thm}[Theorem \ref{sec4:thm-1}] \label{add1}Let $f$ be a non-degenerate quasi-homogeneous polynomial, then
\begin{equation}
\Theta^1_f(s)\equiv 0
\end{equation}
\end{thm}

We have obtained several interesting properties of the new torsion type invariants. There are still many questions arising from this paper. It is interesting to see the variation of our new torsion type invariant under some compact perturbation of metrics. The technical condition $q_M-q_m<1/3$ in Theorem \ref{sect-1-theo-1} may be not necessary for the convergence estimate. We hope that our result can be extended to more generalized LG system $(M,g,f)$. Finally, we expect the torsion invariant defined (or its "equivariant version") can be identified to the BCOV torsion by LG/CY correspondence. 

The rest of the paper is organized as follows: In Section 2, we review some background on analytic torsion and strongly tame systems; In Section 3, we review some preliminary results; In Section 4, we study the special case of harmonic oscillator; In Section 5, we study the short time heat kernel expansion; In Section  6, we discuss the trace class property for the heat kernel and establish Theorem~\ref{add1}; In Section 7, we prove the index theorem Theorem~\ref{add2} and some results related to the new torsion type invariants.

\section{Background }
Let $(M^n, g)$ be a closed oriented $n$-dimensional Riemannian manifold. Given a representation $\rho:\pi_1(M)\to O(r)$, there is the associated flat bundle $E_\rho\to M$ with metric $h$. There is the de Rham complex valued on $E_\rho$:
$$
\cdots\to \Omega^{p-1}(M,E_\rho)\xrightarrow{d^{p-1}} \Omega^p(M,E_\rho)\xrightarrow{d^p}\Omega^{p+1}(M,E_\rho)\to \cdots.
$$
By means of the metrics on $M$ and on $E_\rho$, we can define the adjoint operator $(d^p)^\dag$ and the Laplacian operator acting on $p$-forms:
$$
\Delta_p=(d^{p-1})(d^{p-1})^\dag+(d^p)^\dag d^p. 
$$
$\Delta_p$ can be viewed as a unbounded self-adjoint operator acting on the $L^2$ space $L^2\Omega^p(M,E_\rho)$. 
Since $M$ is compact, $\Delta_p$ has purely discrete spectrum. This induces the de Rham-Hodge theorem: 
\begin{align*}
&L^2\Omega^p(M,E_\rho)=\Hc^p\oplus \im d^{p-1}\oplus \im (d^p)^\dag\\
&\Hc^p\cong H^p_{\dR}(M,E_\rho),
\end{align*}
where $\Hc^p:=\{\phi\in \Omega^p|\Delta_p \phi=0\}$. The operator $e^{-t\Delta_p}:L^2\Omega^p\to L^2\Omega^p$ forms a strongly continuous semigroup and there exits the heat kernel expansion as $t\to 0$: 
\begin{equation}
\Tr e^{-t\Delta_p}=t^{-\frac{n}{2}}(c^p_1+c^p_2 t+c^p_3 t^2+\cdots).
\end{equation}
The constants $c^p_i,i=1,2,\cdots$ are called spectral invariants of $(M,g,\rho)$. Let $\Pi: L^2\Omega^*\to \Hc^*$ be the projection operator. 
The zeta function $\zeta_p(s)$ of $\Delta_p$ is defined as 
$$
\zeta_p(s)=\frac{1}{\Gamma(s)}\int^\infty_0 \Tr(e^{-t\Delta_p}-\Pi)t^{s-1}dt
$$
By heat kernel expansion, if $\re(s) >\frac{n}{2}$, then $\zeta_p(s)$ is an analytic function and actually there is $\zeta_p(s)=\sum_{\lambda^p_i>0}(\lambda_i^p)^{-s}$, where $\lambda^p_i$ are the spectrum of $\Delta_p$. By a result of Seeley \cite{Se}, $\zeta_p(s)$ can be extended meromorphically to $\C$ such that it is analytical at $s=0$. Then the Ray-Singer torsion is defined as 
\begin{equation*}
\log T_{M,\rho}:=\frac{1}{2}\sum_{p=0}^n (-1)^p p\zeta'_p(0). 
\end{equation*}
$T_{M\rho}$ was defined in \cite{RS1}, and was proved to have such properties: (1) $T_{M,\rho}$ is independent of the metric $g$;(2)If $M$ is even, then $T_{M,\rho}$ is constant $1$. The definition of the analytic torsion by Ray-Singer imitates the combinatory definition of The Reidemeister-Franz torsion $\tau_{M,\rho}$, which was defined much earlier (see \cite{Re, Fr, Mi}). $\tau_{M,\rho}$ was observed to be the first topological invariant which can distinguish two homotopy manifolds. The famous Ray-Singer conjecture says that the two invariants are identical. This conjecture was proved by Cheeger and Mueller \cite{Ch, Mu1}. Mueller (\cite{Mu2})and Bimut-Zhang (\cite{BZ}) extended this equivalence theorem to torsions with general coefficient rings.  Afterwards, the torsion invariants have been defined and studied on more generalized situations (see \cite{BGS, BFKM, DM, Ma1, Ma2, MW} and references there).  

Ray-Singer \cite{RS2} generalized their definition of torsion invariants to the $\bpat$-operator on complex manifold. Let $M$ be a complex $n$-dimensional hermitian manifold and let $L\to M$ be a holomorphic flat bundle on $M$ with hermitian metric. We have Dolbeault complex 
$$
\cdots\to \Omega^{p,q-1}(M,L)\xrightarrow{\bpat}\Omega^{p,q}(M,L)\xrightarrow{\bpat}\Omega^{p,q+1}(M,L)\to \cdots
$$
and the Dolbeault cohomology $H^{p,q}_{\bpat}(M,L)=\ker\bpat/\im \bpat$. By using the metrics on $M$ and $L$, one can define the adjoint operator $\bpat^\dag$ and the complex Laplacian $\Delta_{p,q}=\bpat\bpat^\dag+\bpat^\dag \bpat$ from $\Omega^{p,q}(M,L)$ to itself. If $M$ is K\"ahler, $\Delta_\bpat$ is half of the real Laplacian. Similarly, one can define the zeta function $\zeta_{p,q}(s)$ of $\Delta_{p,q}$.  Fix $p=0,1,\cdots,$ the $\bpat$-torsion $T_p(M,\rho)$ is defined as 
$$
\log T_p(M,\rho):=\frac{1}{2}\sum_{q=0}^n(-1)^q q\zeta'_{p,q}(0).
$$
Now assuming that $(M,\om)$ is a K\"ahler manifold with K\"ahler form $\om$. In \cite{RS2}, Ray and Singer proved the following properties:
\begin{enumerate}
\item If $\rho_1,\rho_2$ are two representations, then $T_p{M,\rho_1}/T_p(M,\rho_2)$ is independent of the choice $\om$ in the same cohomology class. 
\item $T_p(M,\rho)$ may depend on the K\"ahler class $[\om]$. 
\item There is the vanishing theorem: $\sum_{p=0}^n (-1)^p\log T_p(M,\rho)=0$.
\end{enumerate}

The above definition of zeta function can be simplified via the concepts of superspace and supertrace.  Let $N:\Omega^p\to \Omega^p$ be the number operator given by $N(\phi^p)=p\phi^p$. $\Lambda^*(\C^n)=\oplus_k \Lambda^k(\C^n)$ is a superspace and if $A\in \End(\Lambda^*(C^n))$, then the superstrace of $A$ is defined as 
$$\str A=\sum_{k=1}^{2n} (-1)^{k} \tr A|_{\Lambda^k(\C^{n})}.$$
Then the zeta function for a real manifold (or complex manifold for $\bpat$-operator) can be rewritten as 
$$
\zeta(s)=\frac{1}{\Gamma(s)}\int^\infty_0 \Str(N(e^{-t\Delta}-\Pi))t^{s-1}dt,
$$
and the torsion is defined as 
$$
T_{M,\rho}:=e^{-\zeta'(0)}.
$$
Here we use "$\tr,\str$" to denote pointwise trace and supertrace and use "$\Tr, \Str$" to denote corresponding concepts in Hilbert spaces, respectively. 

The vanishing theorem for $\bpat$-operator is equivalent to the fact that $\zeta(s)\equiv 0$. Unlike the real case, it is less obvious to extract information of the holomorphic structure from the $\bpat$-torsion.  However, physicists have found the new torsion invariants when studying the quantum field theory and K\"ahler gravity. 

In~\cite{BCOV1,BCOV2}, Bershadsky, Cecotti, OOguri and Vafa have defined a new zeta function, which we call the BCOV zeta function, as
$$
\zeta_{\BCOV}(s):=\frac{1}{\Gamma(s)}\int^\infty_0 \Str(pq(e^{-t\Delta}-\Pi))t^{s-1}dt,
$$
and the corresponding BCOV-torsion is defined as 
$$
T_{M,\rho}^{\BCOV}:=e^{-\zeta'_{\BCOV}(0)}.
$$
Set $\rho$ be the trivial representation and let $M$ be a CY manifold, then in physics, $T_M^{\BCOV}$ is interpreted as the genus $1$ partition function of the Topological field theory on $M$.  

We remark that, considering the vanishing theorem for the torsion for complex manifolds, it is easy to see that
$$
\zeta_{\BCOV}(s):=\frac{1}{2\Gamma(s)}\int^\infty_0 \Str(N^{2}(e^{-t\Delta}-\Pi))t^{s-1}dt.
$$
Thus, the BCOV invariant can be defined independent of the complex structures.

BCOV theory corresponds to the Kodaira-Spencer deformation theory of complex structures on K\"ahler manifolds. It works especially well for CY manifolds. See~\cite{FLY}.  It gives the genus $0$ and genus $1$ theory of CY B model in mirror symmetry phenomenon. The genus $0$ part was given by Kontsevich-Barannikov construction \cite{BK}, The genus $1$ invariant has been defined rigorously by the second-named auther, Z. Lu and K. Yoshikawa~\cite{FLY}. Fang-Lu-Yoshikawa \cite{FLY} and A. Zinger \cite{Zi} have also solved a conjecture of Bershadsky-Cecotti-Ooguri-Vafa about the quintic threefold. The discussion of torsion invariants and modular forms on $K3$ surface can be found in \cite{Yo1} and the subsequent papers \cite{Yo2}. An progress about the construction of the higher genus invariants has been made by K. Castello and S. Li \cite{CL}. 

Given a quintic polynomial $f(\z)=z_1^5+\cdots+z_5^5$, the typical CY 3-fold is given by the zero locus $X_f:=\{[\z]|f(\z)=0\}\subset \IP^4$. However, according to LG/CY correspondence conjecture by physicists, the quantum information of $X_f$ can also be read out from the Landau-Ginzburg model $(\C^5,f)$. The CY and LG models form the known global mirror symmetry picture (ref. \cite{CR}).  The realization of LG model in mathematics is the singularity theory. In the global mirror symmetry picture, there is also the LG to LG mirror conjecture (see the recent work \cite{HLSW}). The LG A model has been built by Fan, Jarvis and Ruan based on Witten's early work on $r$-spin curves, and now is called FJRW theory (ref. \cite{FJR2,FJR3, FFJMR}). The LG B model theory treat the universal unfolding of a singularity $(\C^n,f(\z))$. The Frobenius manifold structure was first found by K. Saito \cite{Sa1,Sa2, ST} via the theory of primitive forms and then later by B. Dubrovin \cite{Du} by oscillatory integration.  However, compared to the CY case, the genus $1$ geometrical invariants are absent in singularity theory.  

Motivated by the early work of Cecotti-Vafa \cite{CV}, the first-named author studied the Hodge theory of the twisted Cauchy-Riemann operator $\bpat_f=\bpat+\pat f\wedge$ on the "strongly tame section bundle system $(M,g,f)$". A section bundle system $(M,g, f)$ consists of a complete non-compact K\"ahler manifold $M$ with metric $g$, which has bounded geometry, and a holomorphic function $f$ defined on $M$. 

\begin{df} The section-bundle system is said to be strongly tame, if for any
constant $C>0$, there is
\begin{equation}
|\nabla f|^2-C|\nabla^2 f|\to \infty, \;\text{as}\;d(x, x_0)\to
\infty.
\end{equation}
Here $d(x,x_0)$ is the distance between the point $x$ and the base
point $x_0$.
\end{df}

A famous example of strongly tame system is a quintic polynomial $f$ on the Euclidean space $\C^5$).  For strongly tame section-bundle system $(M,g,f)$, we have the
fundamental theorem. 

\begin{thm}\label{thm:intro-1}\cite[Theorem 1.4]{Fa} Suppose that $(M,g)$ is a K\"ahler manifold with bounded
geometry. If $\{(M.g),f\}$ is a strongly tame section-bundle system,
then the form Laplacian $\Delta_f$ has purely discrete spectrum and
all the eigenforms form a complete basis of the Hilbert space
$L^2(\Lambda^\bullet(M))$.
\end{thm}

This theorem leads to the following:

\begin{crl}\cite[Theorem 2.42 and Theorem 2.49]{Fa}\label{crl:intro-2} The section-bundle systems $(\C^n, i\sum_j dz^j\wedge dz^\jb, W)$
and $((\C^*)^n, i\sum_j \frac{dz^j}{z^j}\wedge
\frac{dz^\jb}{z^\jb},f)$ are strongly tame, if
\begin{itemize}
\item $W$ is a non-degenerate quasi-homogeneous polynomial with
homogeneous weight $1$ and of type $(q_1,\cdots,q_n)$ with all
$q_i\le 1/2$.
\item $f$ is a convenient and non-degenerate Laurent polynomial
defined on the algebraic torus $(\C^*)^n$.
\end{itemize}
Therefore, the corresponding form Laplacian has purely discrete
spectrum and all the eigenforms form a complete basis of the Hilbert
space $L^2(\Lambda^\bullet(M))$.
\end{crl}

By the above result, the first-named author proved in \cite[Theorem 2.52]{Fa} the following Hodge-de Rham theorem:
\begin{align*}
&L^2\Omega^p(M)=\Hc^p\oplus \im \bpat^{p-1}_f\oplus \im (\bpat^p_f)^\dag\\
&\Hc^p\cong H^p_{((2),\bpat_f)}\cong\begin{cases}
0 &\text{if}p\neq n\\
\Omega^n(M)/df\wedge \Omega^{n-1}(M),&\text{if}p=n,
\end{cases}
\end{align*}
where $\Hc^*:=\{\phi\in \Omega^*|\Delta_f \phi=0\}$ and $H^*_{((2),\bpat_f)}$ is the $L^2$ $\bpat_f$-cohomology. This Hodge-de Rham theorem recovers the partial information of the local $\C$-algebra  of the singularity $f$. The further study of the tt*-structure on the Hodge bundle over the deformation parameter space reveals essentially the Frobenius manifold structure of the singularity $f$ (see \cite{Fa}).

\section{Preliminary}

\subsection{Basic notations}

Let $\C^n$ be the standard Euclidean space with complex coordinates $\z=(z_1,\cdots,z_n)$, where $ z_j=x_j+iy_j, j=1,\cdots, n$. Let $\alpha$ be a $2n$-multiple index having the form
$$
\alpha=(\alpha_x,\alpha_y)=(\alpha_1,\cdots, \alpha_n,\alpha_{n+1},\cdots,\alpha_{2n}).
$$
The real derivatives are defined as
$$
D^\alpha_\z=D_x^{\alpha_x}D_y^{\alpha_y}=(\frac{\pat}{\pat x_1})^{\alpha_1}\cdots (\frac{\pat}{\pat y_{2n}})^{\alpha_{2n}}.
$$
If $\alpha$ is a $n$-multiple index, then we have the complex derivatives
$$
\pat^\alpha_\z=(\frac{\pat}{\pat z_1})^{\alpha_1}\cdots (\frac{\pat}{\pat z_{n}})^{\alpha_{n}}
$$
Given a $2n$-multiple index $\alpha=(\alpha_x,\alpha_y)$, we have an induced $n$-multiple index $\alphab:=(\alpha_x+\alpha_y)$  which satisfies $|\alpha|=|\alphab|$. Since
$$
\frac{\pat}{\pat x_i}=\frac{\pat}{\pat z_i}+\frac{\pat }{\pat \zb_i},\;\frac{\pat}{\pat y_i}=\sqrt{-1}(\frac{\pat}{\pat z_i}-\frac{\pat}{\pat \zb_i}),
$$
we have for any holomorphic function $f$ that
$$
D^\alpha_\z |f|^2=\sum_{\alpha_1+\alpha_2=\alpha}\frac{\alpha!}{\alpha_1!\alpha_2!}D^{\alpha_1}_\z f\cdot D^{\alpha_2}_\z \fb=\sum_{\alpha_1+\alpha_2=\alpha}\frac{\alpha!}{\alpha_1!\alpha_2!}(\sqrt{-1})^{|\alpha_{1y}|+3|\alpha_{2y}|}\pat^{\alphab_1}_\z f\cdot \overline{\pat^{\alphab_2}_\z f}.
$$
Define the pointwise norm
$$
|D^\alpha |f|^2|_+:=\sum_{\alpha_1+\alpha_2=\alpha}\frac{\alpha!}{\alpha_1!\alpha_2!}|\pat^{\alphab_1}_\z f|\cdot |\pat^{\alphab_2}_\z f|,
$$
then we have
\begin{equation}\label{sec0-ineq1}
|D^\alpha_\z |f|^2|\le |D^\alpha |f|^2|_+.
\end{equation}

Furthermore, if $f$ is a polynomial of the following form:
$$
f(\z)=\sum_l \frac{c_l}{b_l!}z^{b_l}=\sum_l \frac{c_l}{b_{l1}!\cdots b_{ln}!}z_1^{b_{l1}}\cdots z_n^{b_{ln}},
$$
then we have
\begin{align}\label{sec0-esti3}
|\pat f|^2&=\sum_i |\pat_i f|^2=\sum_i \left|\sum_l \frac{c_l}{(b_l-e_i)!}z^{b_l-e_i} \right|^2\nonumber\le \sum_i \left|\sum_l \frac{c_l}{(b_l-e_i)!}|z|^{|b_l|-1} \right|^2\\
&\le \sum_i |\sum_l \frac{c_l}{(b_l-e_i)!}|^2 (|z|^{2\max|b_l|}+1)=C(|z|^{2\max|b_l|}+1)
\end{align}

\subsection{Clifford operators and supertrace}

Let $\iota_{\pat_{z_i}}$ be the contraction operator such that $\iota_{\partial z^{i}}(dz^{i})=1$. Then we define the Clifford operators
\begin{align*}
&c_{i}=dz^{i}\wedge-\iota_{\pat_{z_i}},\quad \hc_i=dz^{i}\wedge+\iota_{\pat_{z_i}}\\
&\bc_i=d\zb^i\wedge-\iota_{\pat_{\zb_i}},\quad \hbc_i=d\zb^i\wedge+\iota_{\pat_{\zb_i}}
\end{align*}
Conversely, we have the expression
\begin{align*}
&dz^i\wedge=\frac{1}{2}(c_i+\hc_i),\quad \iota_{\pat_{\z_i}}=\frac{1}{2}(\hc_i-c_i)\\
&d\zb^i\wedge=\frac{1}{2}(\bc_i+\hbc_i),\quad\iota_{\pat_{\zb_i}}=\frac{1}{2}(\hbc_i-\bc_i)
\end{align*}

We have 
$$
c_{i}^{2}=\bc_{i}^{2}=-1,\quad\hc_{i}^{2}=\hbc_{i}^{2}=1.
$$ 
Otherwise these elements are anti-commutative. $\End(\Lambda^{*}\C^{n})$ is then generated by the above elements.
\begin{df}
The number operator $N$ is defined as $N\alpha =k\alpha$, when $\alpha\in \Lambda^{k}(\C^{n})$. 
\end{df}
It is not hard to see that 
$$
N= n+{1\over 2}\sum_{i=1}^{n}(c_{i}\hc_{i}+\bc_{i}\hbc_{i}).
$$

\begin{df} Let $A\in \End(\Lambda^*\C^n)$, then the super trace of $A$ on $\Lambda^{*}(\C^{n})$ is defined as 
$$\str A=\sum_{k=1}^{2n} (-1)^{k} \Tr A|_{\Lambda*(\C^{n})}.$$
\end{df}

The Clifford algebra structure of $\End(\Lambda T^{*}\C^{n})$ easily induces the following result

\begin{lm}\label{sec2:lm-0}
\begin{equation}
\str [\Pi_{i=1}^{n}(c_{i}\hc_{i}\bc_{i}\hbc_{i})] =4^{n}.
\end{equation} 
Furthermore, all other monomials of $c_{*}, \hc_{*},\bc_{*}$ and $\hbc_{*}$ has super-trace 0.
\end{lm}

Notice that $\Delta_f$ has local expression: 
\begin{equation}
\Delta_f=\Delta_{\bpat}+L_f+|\nabla f|^2,
\end{equation}
where 
$$
L_f=-(g^{\mub\nu}\nabla_\nu f_l \iota_{\pat_{\mub}}dz^l\wedge+\overline{g^{\mub\nu}\nabla_\nu f_l \iota_{\pat_{\mub}}dz^l\wedge}).
$$
Here we have used the Einstein summation convention with respect to $\mu,\nu,l$. 

\begin{prop}\label{sec1:prop-1} Let $m\in \N$, we have
\begin{equation}
\str L_f^m=
\begin{cases}
0,&\text{if}\;1\le m<2n\\
\frac{(2n)!}{2^{2n}}(-1)^n 4^n |\det(g)|^{-2}|\det(\pat^2 f)|^2,&\text{if}\;m=2n.
\end{cases}
\end{equation}
\end{prop}
In particular, if $g=\frac{1}{2}\sum_i dz^i d\zb^i$ is the standard hermitian metric, we have 
\begin{equation}
\str L_f^{2n}=(2n)! (-1)^n 4^n |\det(\pat^2 f)|^2.
\end{equation}

\begin{proof} 

Assume without loss of generality that the Hessian $\pat^2f=\nabla \pat f$ of $f$ is nonzero at point $p$. Since $\pat^2 f$ is symmetric, we can choose a $\C$-valued matrix $C$ such that $C^T\cdot \pat^2f(p)\cdot C=I$. Do holomorphic coordinate transformation $z^i=C_{ij} w^j$, then in coordinate system $\{w^i\}$, the holomorphic function $\tilde{f}(\w)=f(\z)=f(C\cdot \w)$ has Hessian at $p$:
$$
\pat^2_\w f=C^T\cdot \pat^2_\z f\cdot C=I.
$$
Assume that the metric tensor $g=i g_{\mu \nub}dz^\mu d z^\nub=i\tilde{g}_{k \lb}dw^k dw^\lb$, then we have the relation
$$
\tilde{g}=C^T\cdot g\cdot \bar{C},\;\text{and}\;\tilde{g}^{-1}=\overline{C^{-1}}\cdot g^{-1}\cdot (C^{-1})^T.
$$
Note that $-g^{\mub\nu}\nabla_\nu f_l \iota_{\pat_{\mub}}dz^l\wedge=[\bpat^\dag,\pat f\wedge]$ is independent of the choice of local coordinate system. Hence in $\{w_i\}$ system, we have 
\begin{align*}
L_f(p)&=-(\tilde{g}^{\kb l} \iota_{\pat_{\kb}}dw^l\wedge+\overline{\tilde{g}^{\kb l} \iota_{\pat_{\kb}}dw^l\wedge})\\
&=-\tilde{g}^{\kb l} (\iota_{\pat_{\kb}}d w^l\wedge+\iota_{\pat_l}d w^\kb\wedge )\\
&=-\tilde{g}^{\kb l}\left( \frac{1}{2}(\hbc_k-\bc_k)\frac{1}{2}(c_l+\hc_l)+\frac{1}{2}(\hc_l-c_l)\frac{1}{2}(\hbc_k+\bc_k)\right)\\
&=-\frac{1}{4}\tilde{g}^{\kb l}\left(\hbc_k c_l+\hbc_k \hc_l-\bc_k c_l-\bc_k \hc_l+\hc_l \hbc_k+\hc_l \bc_k-c_l \hbc_k-c_l \bc_k \right)\\
&=\frac{1}{2}\tilde{g}^{\kb l}\left(\bc_k \hc_l-\hbc_k c_l \right)
\end{align*}
Define 
$$
A_k=\tilde{g}^{\kb l}\hc_l,\; B_l=\tilde{g}^{\kb l}c_l,
$$
then 
$$
L_f^m=\frac{1}{2^m}\{\bc_1A_1+\cdots+\bc_n A_n+B_1\hbc_1+\cdots+B_n\hbc_n\}^{m}.
$$
Note that if $x_1,\cdots,x_a$ are commutative quantities, then we have the expansion
$$
(x_1+\cdots+x_a)^m=\sum_{|\alpha|=m}\frac{m!}{\alpha !}x^\alpha,
$$
where $\alpha=(\alpha_1,\cdots,\alpha_a)$ and $x^\alpha=\prod_{i=1}^a x_i^{\alpha_i}$.
Hence for $1\le m<2n$, by Lemma \ref{sec2:lm-0}, we have pointwise identity
$$
\str L_f^m=0.
$$
If $m=2n$, then we have 
\begin{align*}
\str L_f^{2n}&=\frac{(2n)!}{2^{2n}}\str (\bc_1A_1\circ\cdots\bc_n A_n\circ\cdots B_1\hbc_1\circ\cdots B_n\hbc_n),
\end{align*}
Since all the terms $\bc_\mu A_\mu$ and $B_\nu\hbc_\nu$ commute mutually. 

So we have
\begin{align*}
\str L_f^{2n}(p)&=\frac{(2n)!}{2^{2n}}\str \left(\bc_1\circ\cdots\circ\hbc_n\circ A_1\cdots A_nB_1\cdots B_n\right)\\
&=\frac{(2n)!}{2^{2n}} \str \left( \bc_1\circ\cdots\circ\hbc_n\circ \tilde{g}^{\bar{1} l_1}\hc_{l_1}\cdots \tilde{g}^{\bar{n} l_n}\hc_{l_n}\circ \tilde{g}^{\bar{1} k_1}c_{k_1}\cdots \tilde{g}^{\bar{n} k_n} c_{k_n}\right)\\
&=\frac{(2n)!}{2^{2n}}(-1)^n \str \left( \bc_1\circ\cdots\circ\hbc_n\circ c_1\cdots \hc_n (-1)^{\sgn(l_1,\cdots,l_n)}\tilde{g}^{\bar{1}l_1}\cdots \tilde{g}^{\bar{n}l_n}(-1)^{\sgn(k_1,\cdots,k_n)}\tilde{g}^{\bar{1}k_1}\cdots \tilde{g}^{\bar{n} k_n} \right)\\
&=\frac{(2n)!}{2^{2n}}(-1)^n 4^n |\det \tilde{g}|^{-2}\\
&=\frac{(2n)!}{2^{2n}}(-1)^n 4^n|\det g|^{-2} |\det \pat^2 f|^2(p)
\end{align*}
\end{proof}

\subsection{Estimates of potential function and standard heat kernel }

\begin{df}\label{df:qhomPoly}
A \emph{quasi-homogeneous} (or \emph{weighted homogeneous})
\emph{polynomial} $f \in \mathbb{C} [z_1, \dots, z_n]$ is a polynomial for which there
exist positive rational numbers $q_1, \dots, q_n \in \Q^{>0}$, such
that for any $\lambda \in \mathbb{C}^*$

\begin{equation}\label{eq:qhomocharge}
f(\lambda^{q_1}z_1, \dots, \lambda^{q_n}z_n) =\lambda f(z_1,
\dots, z_n).
\end{equation}
 We will call $q_j$ the \emph{weight} of $z_j$. We define
$d$ and $k_i$ for $i\in \{1,\dots,n\}$ to be the unique positive
integers such that $(q_1,\dots,q_n) = (k_1/d, \dots, k_n/d)$ with
$\gcd(d,k_1,\dots,k_n) =1$.

\end{df}

Throughout this paper we will need a certain non-degeneracy condition on $f$.

\begin{df}\label{df:nondegenerate} We call $f$ \emph{nondegenerate}
if \begin{enumerate}
\item\label{it:nondegen-unique}
$f$ contains no monomial of the form $z_iz_j$, for $i\neq j$ and
\item The hypersurface defined by $f$ in weighted
projective space is non-singular, or, equivalently, the affine
hypersurface defined by $f$ has an isolated singularity at the
origin.
  \end{enumerate}
    \end{df}

\begin{prop}\cite[Theorem 3.7(b)]{HK}.
If $f$ is a non-degenerate, quasi-homogeneous polynomial, then the weights
   $q_i$ are bounded by $q_i\leq \frac{1}{2}$ and are unique.
   \end{prop}

For nondegenerate quasihomogeneous polynomials, we have the following important estimates:

\begin{prop}\cite[Theorem 5.7]{FJR1}\label{prop-esti1}

Let $f \in {\C}[z_1, \dots, z_n]$ be a non-degenerate,
quasi-homogeneous polynomial. Then for any $n$-tuple $(u_1,
\dots, u_n) \in {\C}^n$ we have
\[ |u_i| \leq C \left(\sum^n_{i=1}\left|\frac{\partial f}{\partial z_i}(u_1, \dots,
u_n)\right|+1 \right)^{\gamma_i},\] where
$\gamma_i=\frac{q_i}{\min_j(1-q_j)}\le 1$ for all $ i \in \{1, \dots, n\}$ and the constant $C$ depends
only on $f$. If $q_i<1/2$ for all $ i \in \{1, \dots, n\}$, then $\gamma_i<1$ for
all $ i \in \{1, \dots, n\}$.
\end{prop}

\begin{crl}\label{sect-2-lemm-2.9} If $f \in {\C}[z_1, \dots, z_n]$ is a non-degenerate,
quasi-homogeneous polynomial. Then there exists a constant $C$ depending only on $f$ such that
\begin{equation}
|\pat f|^2\ge \frac{1}{C}|\z|^2-1.
\end{equation}
\end{crl}

Define the potential function:
\begin{equation}
V(\z):=|\pat f|^2=\sum_{i=1}^n|\pat_i f|^2.
\end{equation}

Set the following quantities
\begin{align*}
&\q=(q_1,\cdots,q_n),\;q_M:=\max\{q_i\},\;q_m:=\min\{q_i\}\\
&\delta_M=\frac{1}{2(1-q_M)},\;\delta_0=\delta_M q_m.
\end{align*}
It is obvious that
\begin{equation}
q_m,q_M\le \frac{1}{2},\delta_M\le 1,\delta_0\le \frac{1}{2}.
\end{equation}

Let
$$
h=\sum_l \frac{c_l}{b_l!}z^{b_l}=\sum_l \frac{c_l}{b_{l1}!\cdots b_{ln}!}z_1^{b_{l1}}\cdots z_n^{b_{ln}}
$$
be a quasihomogeneous polynomial with charge $c_h$, i.e., for any $\lambda\in \C^*$, there is
$$
h(\lambda_1^{q_1}z_1,\cdots,\lambda_n^{q_n} z_n)=\lambda^{c_h} h(z_1,\cdots,z_n).
$$
If $f$ is a nondegenerate quasihomogeneous polynomial, then by Proposition \ref{prop-esti1} we have
$$
|z_i| \leq C \left(\sum^n_{i=1}\left|\frac{\partial f}{\partial z_i}(z_1, \dots,
z_n)\right|+1 \right)^{\gamma_i}.
$$
Hence for multiple index $\alpha$,
\begin{align*}
&\prod^n_{i=1}|z_i|^{b_{li}-\alpha_i}\le \prod^n_{i=1}C^{b_{li}-\alpha_i}(\sum^n_{i=1}|\frac{\partial f}{\partial z_i}(z_1, \dots,
z_n)|+1 )^{(b_{li}-\alpha_i)\gamma_i}\\
&\le C^{\sum_i b_{li}-|\alpha|}(\sum^n_{i=1}|\frac{\partial f}{\partial z_i}(z_1, \dots,
z_n)|+1 )^{\frac{c_h-\alpha\cdot \q}{\min_j(1-q_j)}}\\
&\le C^{\sum_i b_{li}-|\alpha|}(n+1)^{c_h-\alpha\cdot \q}(V+1)^{c_h-\alpha\cdot \q}.
\end{align*}

This shows that
\begin{align}
|\pat^\alpha h(\z)|&\le \sum_l \frac{c_l}{(b_l-\alpha)!}|z^{b_l-\alpha}|\le \sum_l \frac{c_l}{(b_l-\alpha)!}C^{\sum_i b_{li}-|\alpha|}(n+1)^{c_h-\alpha\cdot \q}(V+1)^{c_h-\alpha\cdot \q}\nonumber\\
&=C_\alpha (V+1)^{c_h-\alpha\cdot \q}.
\end{align}
By this inequality, we have
\begin{align*}
|D^\alpha|\pat h|^2|_+&=\sum_{\alpha_1+\alpha_2=\alpha}\frac{\alpha!}{\alpha_1!\alpha_2!}|\pat^{\alphab_1}\pat_i h||\pat^{\alphab_2} \pat_i h|\\
&\le \sum_{\alpha_1+\alpha_2=\alpha}\frac{\alpha!}{\alpha_1!\alpha_2!}C_{\alphab_1}(V+1)^{c_h-\alphab_1\cdot \q-q_i}C_{\alphab_2}(V+1)^{c_h-\alphab_2\cdot \q-q_i}\\
&\le C_\alpha (V+1)^{2c_h-|\alpha| q_m-2q_m}.
\end{align*}
In particular, we have
\begin{equation}\label{sec0-esti2}
|D^\alpha|\pat f|^2|_+\le C_\alpha (V+1)^{2-|\alpha| q_m-2q_m}.
\end{equation}

Denote by
$$
\z_t=(\z^t_1,\cdots,\z^t_n),\;z^t_i:=t^{\delta_M q_i}z_i.
$$

\begin{lm}\label{sect-2-Lemm-1} If $1-\q\cdot \alphab\ge 0$, then we have
\begin{equation}
\quad\pat^\alphab_\z f(z_1,\cdots,z_n)=t^{-\delta_M(1-\q\cdot\alpha)}(\pat^\alphab_\z f)(z^t_1,\cdots,z^t_n).
\end{equation}
In particular, we have
\begin{align*}
(i)&\quad |\pat_i\pat_j f|(z_1,\cdots,z_n)=t^{-\delta_M(1-q_i-q_j)}|\pat_i\pat_j f|(z^t_1,\cdots,z^t_n)\\
(ii)&\quad |\pat_i f|^2(z_1,\cdots,z_n)=t^{-2\delta_M(1-q_i)} |\pat_i f|^2(z^t_1,\cdots,z^t_n).
\end{align*}
\end{lm}

\begin{proof} The results are due to the homogeneity
 \begin{equation}
 (\pat^\alphab_\z f)(\lambda^{q_1}z_1,\cdots,\lambda^{q_n}z_n)=\lambda^{1-\q\cdot \alphab} (\pat^\alphab_\z f)(z_1,\cdots,z_n).
 \end{equation}
 and set $\lambda=t^{\delta_M}$.
\end{proof}

\begin{lm}\label{sec1-lma1}Fix any $T\ge 1$. We have the estimate of $V$:
\begin{align*}
(i)&\quad V(z^t_1,\cdots,z^t_n)\le t V(z_1,\cdots,z_n),\;\forall t\in (0,1]\\
(ii)&\quad |\pat^\alphab_\z V|(z_1,\cdots,z_n)\le t^{\delta_0 (|\alphab|+2)-2\delta_M} T^{|\alphab|+2} |\pat^\alphab_\z V|(z^t_1,\cdots,z^t_n),\;\forall t\in (0,T]\\
(iii)&\quad |D^\alpha_\z V|(z_1,\cdots,z_n)\le t^{\delta_0(|\alpha|+2)-2\delta_M} T^{|\alpha|+2} |D^\alpha_\z V|_+(z^t_1,\cdots,z^t_n).\;\forall t\in (0,T]
\end{align*}

In particular, we have for any $t\in (0,T]$:
\begin{align*}
(iv)&\quad t^{\frac{3}{2}}|\nabla V|(z_1,\cdots,z_n)\le t^{\delta_3} T^3|\nabla V|(z^t_1,\cdots,z^t_n)\\
(v)&\quad t^2 |\Delta V|(z_1,\cdots,z_n)\le t^{2\delta_2} T^4 |\Delta V|(z^t_1,\cdots,z^t_n),
\end{align*}
where
\begin{equation}
\delta_2=\frac{1-2(q_M-q_m)}{2(1-q_M)},\;\delta_3=\frac{1-3(q_M-q_m)}{2(1-q_M)}
\end{equation}
\end{lm}

\begin{proof}It suffices to prove (iii). By (\ref{sec0-ineq1}), we have
\begin{align*}
|D^\alpha_\z V|(\z)&\le |D^\alpha_\z V|_+(\z)=\sum_i\sum_{\alpha_1+\alpha_2=\alpha}\frac{\alpha!}{\alpha_1!\alpha_2!}|\pat^{\alphab_1}_\z \pat_i f|\cdot |\pat^{\alphab_2}_\z \pat_i f|\\
&\le \sum_{i}t^{\delta_M \q\cdot\alphab-2\delta_M(1-q_i)}\sum_{\alpha_1+\alpha_2=\alpha}\frac{\alpha!}{\alpha_1!\alpha_2!}|\pat^{\alphab_1}_\z \pat_i f|(\z_t)\cdot |\pat^{\alphab_2}_\z \pat_i f|(\z_t)\\
&\le \sum_{i}t^{\delta_M q_m|\alpha|-2\delta_M(1-q_m)}t^{\delta_M(\q\cdot\alphab-q_m|\alpha|)+2\delta_M(q_i-q_m)}\sum_{\alpha_1+\alpha_2=\alpha}\frac{\alpha!}{\alpha_1!\alpha_2!}|\pat^{\alphab_1}_\z \pat_i f|(\z_t)\cdot |\pat^{\alphab_2}_\z \pat_i f|(\z_t)\\
&\le t^{\delta_0|\alpha|-2\delta_M(1-q_m)}T^{(|\alpha|+2)(\delta_M(q_M-q_m))}|D^\alpha_\z V|_+(\z_t)\\
&\le t^{\delta_0|\alpha|-2\delta_M(1-q_m)}T^{(|\alpha|+2)}|D^\alpha_\z V|_+(\z_t)\end{align*}
\end{proof}

\subsubsection*{\underline{Mean value estimate}}

Consider the mean value of $V$ between two points $\z$ and $\w$:
$$
g(\z,\w):=\int^1_0 V(\tau(\z-\w)+\w)d\tau.
$$
It is obvious that $g(\z,\w)$ is symmetric about $\z$ and $\w$ and satisfies the following equality
\begin{equation}
(\z-\w)\cdot \nabla_\z g+g(\z,\w)=V(\z),
\end{equation}
where $\z$ is identified with the real vector $(x_1,\cdots,x_n,y_1,\cdots,y_n)$ and $\nabla_\z=(\nabla_{x_1},\cdots,\nabla_{y_n})$ is the gradient operator.

If $f$ is a function or a matrix-valued function, which have the expression $f(\z,\w)=f(\z-\w,\w)$, then we define the notation:
\begin{equation}
\{f(\tau(\z-\w),\w)\}_\tau=\max_{\tau\in [0,1]}||f(\tau(\z-\w),\w)||.
\end{equation}

\begin{lm}\label{sec1-lmb1} Fix any $T\ge 1$. We have the estimate of $g$:
\begin{align*}
(1)&\quad|D^\alpha_\z g|(\z,\w)\le t^{\delta_0|\alpha|-2\delta_M(1-q_m)} T^{|\alpha|+2} \{|D^\alpha_\z V|_+ (\tau(\z_t-\w_t),\w_t)\}_\tau,\;\forall t\in (0,T].
\end{align*}
In particular, we have
\begin{align*}
(2)&\quad t^{\frac{3}{2}}|\nabla_\z g|(\z,\w)\le t^{\delta_3} T^3\{|\nabla_\z V|_+(\tau(\z_t-\w_t),\w_t)\}_\tau\\
(3)&\quad t^2 |\Delta_\z g|(\z,\w)\le t^{2\delta_2} T^4 \{|\Delta_\z V|_+(\tau(\z_t-\w_t),\w_t)\}_\tau,
\end{align*}
\end{lm}

Let
\begin{equation}
\E_1(\z.\w,t)=e^{-tg(\z,\w)},
\end{equation}
then for any multiple index $\alpha$ and any $l\in \N$, we have identities
\begin{align}
&D^l_t \E_1=(-g)^l \E_1\\
&D^\alpha_\z \E_1=G^\alpha(\z,\w,t)\E_1,
\end{align}
where $G^\alpha$ consisting of the derivatives of $g$ has the form $G^\alpha(\z-\w,\w,t)$ and are defined inductively by
\begin{align*}
&G^{e_i}=-tD_{z_i}g\\
&G^{\alpha+e_i}=D_{z_i}G^\alpha+(-tD_{z_i}g)G^\alpha,
\end{align*}
where $e_i=(0,\cdots,0,1,0,\cdots,0)$ is the standard $i$-th unit vector in $\R^{2n}$.

Note that each term in the summation of $G^\alpha$ has the form
$$
(-t)^s D^{\alpha_1}g\cdots D^{\alpha_s}g,
$$
where $\alpha_i$ are multiple indices such that $\alpha_1+\cdots+\alpha_s=\alpha$. Each such term has an upper bound:
$$
t^s |D^{\alpha_1}g\cdots D^{\alpha_s}g|\le t^s\{|D^{\alpha_1} V|_+\}_\tau\cdots\{|D^{\alpha_s} V|_+\}_\tau.
$$
We define a "norm" $F^\alpha_g(\z,\w,t)$ related to $G^\alpha$ to be the corresponding summation of the terms
$t^s\{|D^{\alpha_1} V|_+\}_\tau\cdots\{|D^{\alpha_s} V|_+\}_\tau$.

\begin{lm}\label{sect-2-lemm-2.13} For any $t\in (0,T]$, there is
\begin{equation}
|G^\alpha|(\z,\w, t)\le F^\alpha_g(\z,\w,t)\le t^{(-\frac{1}{2}+\delta_3)|\alpha|}T^{3|\alpha|}F^\alpha_g(\z_t,\w_t,1).
\end{equation}
\end{lm}

\begin{proof}It suffices to prove the second inequality. Each term of $F^\alpha_g$ has the form
$$
t^s\{|D^{\alpha_1} V|_+\}_\tau\cdots\{|D^{\alpha_s} V|_+\}_\tau,
$$
where $\alpha_1+\cdots+\alpha_s=\alpha$.
By Lemma \ref{sec1-lma1}, we have estimate
\begin{align*}
&t^s\{|D^{\alpha_1} V|_+\}_\tau\cdots\{|D^{\alpha_s} V|_+\}_\tau(\z,\w)\\
&\le t^s\cdot t^{\delta_0|\alpha_1|-2\delta_M(1-q_m)}T^{|\alpha|+2}\cdots t^{\delta_0|\alpha_s|-2\delta_M(1-q_m)}T^{|\alpha_s|+2}\{|D^{\alpha_1} V|_+\}_\tau\cdots\{|D^{\alpha_s} V|_+\}_\tau(\z_t,\w_t)\\
&\le t^{\delta_0|\alpha|-s(\delta_M-\delta_2)}T^{|\alpha|+2s}\{|D^{\alpha_1} V|_+\}_\tau\cdots\{|D^{\alpha_s} V|_+\}_\tau(\z_t,\w_t)\\
&\le t^{|\alpha|(\delta_0+\delta_2-\delta_M)}T^{|\alpha|+2s+(|\alpha|-s)(\delta_M-\delta_2)}\{|D^{\alpha_1} V|_+\}_\tau\cdots\{|D^{\alpha_s} V|_+\}_\tau(\z_t,\w_t)\\
&\le t^{|\alpha|(-\frac{1}{2}+\delta_3)}T^{3|\alpha|}\{|D^{\alpha_1} V|_+\}_\tau\cdots\{|D^{\alpha_s} V|_+\}_\tau(\z_t,\w_t).
\end{align*}
Here we used the fact that $\delta_0+\delta_2-\delta_M=-1/2+\delta_3$ and $0\le\delta_M-\delta_2\le 1$. So we proved the lemma.
\end{proof}

\subsubsection*{\underline{Estimate of the standard heat kernel}}

It is known that $\pat_t-\Delta$ has the heat kernel on $\C^n$:
\begin{equation}
\E_0(\z,\w,t)=(4\pi t)^{-n}e^{-\frac{|\z-\w|^2}{4t}}.
\end{equation}
Write $z_j=x_j+iy_j$ and $w_j=u_j+iv_j$. Let $D_{z_i}$ represent the derivatives $D_{x_i}$ or $D_{y_i}$ while $z_i-w_i$ is understood in this case as $x_i-u_i$ or $y_i-v_i$. Then we have
 \begin{align*}
&D_{z_i} \E_0=(-\frac{z_i-w_i}{2t})\E_0\\
&D_t\E_0=(-\frac{n}{t}+\frac{|\z-\w|^2}{4t^2})\E_0,
\end{align*}
Furthermore, given multiple index $\alpha$ and any $l\in \N$, we have identities
\begin{align}
&D^\alpha_\z \E_0=t^{-|\alpha|}P^{|\alpha|}_0(\z-\w,t)\E_0\\
&D^l_t \E_0=t^{-2l}P_1^{2l}(|\z-\w|,t)\E_0.
\end{align}
Here $P^{|\alpha|}_0(\z-\w,t)$ is a polynomial of variables $\z-\w, t$ and its degree with respect to $\z-\w$ is $|\alpha|$ and the higher degree term has constant coefficient, while $P^{2l}_1(|\z-\w|,t)$ is a polynomial of $|\z-\w|, t$ and its degree with respect to $\z-\w$ is $2l$ and the higher degree term has constant coefficient.

Define $|P^{|\alpha|}_0|_+(\z-\w,t)$ to be "the norm" of $P^{|\alpha|}_0$ which is the summation of the absolute value of each monomial appearing in $P^{|\alpha|}_0$. Similarly, we can define $|P_1^{2l}|_+$.

\begin{lm}\label{sect-2-lemm-2.14} We have estimates
\begin{align}
&|D^\alpha_\z \E_0|(\z,\w)\le t^{-(1+\delta_Mq_M)|\alpha|}T^{|\alpha|}|P^{|\alpha|}_0|_+(\z_t-\w_t,1)\E_0(\z,\w)\\
&|D^l_t \E_0|(\z,\w)\le t^{-2(1+\delta_Mq_M)l}T^{2l}|P^{2l}_0|_+(\z_t-\w_t,1)\E_0(\z,\w)
\end{align}
\end{lm}

\begin{proof}It suffices to check the highest term $|\z-\w|^{|\alpha|}$. We have
\begin{align*}
|\z-\w|^{|\alpha|}&=(\sum_i |\z_i-\w_i|^2)^{\frac{|\alpha|}{2}}=(\sum_i |\z^t_i-\w^t_i|^2 t^{-2\delta_Mq_i})^{\frac{|\alpha|}{2}}\\
&\le t^{-\delta_Mq_M|\alpha|}T^{|\alpha|}|\z_t-\w_t|^{|\alpha|}.
\end{align*}
This proves the first inequality. One can prove the second inequality in the same way.
\end{proof}
 
 \section{Heat kernel of complex 1-dimensional harmonic oscillator}
 
 Consider the complex 1-dimensional harmonic oscillator given by the section-bundle system $(\C, \frac{i}{2} dz\otimes d\zb, f=\frac{\tau}{2}z^2)$. The heat kernel of this system will be written down explicitly. 

Notice that $\Delta_f$ has local expression: 
\begin{equation}
\Delta_f=\Delta_{\bpat}-(g^{\mub \nu}\nabla_\nu f_l \iota_{\pat_{\mub}}dz^l\wedge+\overline{g^{\mub \nu}\nabla_\nu f_{\pat_\mub}dz^l\wedge})+|\nabla f|^2.
\end{equation}

Hence we have the formula for this special system. 
\begin{equation}
\Delta_f=-2\pat_z\pat_\zb-2(\tau \iota_{\pat_\zb}dz\wedge +\bar{\tau}\iota_z d\zb\wedge)+2|\tau|^2|z|^2.
\end{equation}

\subsubsection*{\underline{Heat kernel of $1$ forms}}

Let $\varphi=u dz+v d\zb$, then
\begin{align*}
\Delta_f\varphi=2(-\pat_z\pat_\zb u+|\tau|^2|z|^2u+\tau v)dz+2(-\pat_z\pat_\zb v+ |\tau|^2|z|^2 v+\tau u)d\zb.
\end{align*}

The spectrum problem $\Delta_f\varphi=\lambda \varphi$ is transformed to the following problem:
\begin{equation}
\begin{cases}
-\pat_z\pat_\zb u+|\tau|^2|z|^2u+\tau v=\frac{\lambda}{2} u\\
-\pat_z\pat_\zb v+ |\tau|^2|z|^2 v+\tau u=\frac{\lambda}{2} v.
\end{cases}
\end{equation}

We have an evolution $\tau_R$ such that
\begin{equation}
\tau_\R(udz+vd\zb)=(\frac{\tau}{|\tau|}vdz+\frac{\bar{\tau}}{|\tau|}ud\zb).
\end{equation}

The involution commutes with $\Delta_f$ and has two eigenspaces $E_\mp$ with respect to $\mp 1$, where $E_\mp$ are generated by 
$\varphi_\mp=v(\mp \frac{\tau}{|\tau|}dz+d\zb)$.

Consider the restriction of $\Delta_f$ to the space $E_+$, then $\Delta_f \varphi_+=\lambda \varphi_+$ is equivalent to the problem
\begin{equation}
-\pat_z\pat_\zb v+|\tau|^2|z|^2 v+|\tau| v=\frac{\lambda}{2} v.
\end{equation}
On the other hand, when restricted to $E_-$ , the problem of $\Delta_f \varphi_-=\lambda \varphi_-$ is equivalent to the problem
\begin{equation}
-\pat_z\pat_\zb v+|\tau|^2|z|^2 v-|\tau| v=\frac{\lambda}{2} v,
\end{equation}
or equivalently
\begin{equation}
(-\pat_z\pat_\zb v+|\tau|^2|z|^2)v=(\frac{\lambda}{2}+|\tau|)v. 
\end{equation}

Note that $\Delta_0=\pat^2_x+\pat^2_y=4\pat_z\pat_\zb$, the above equation can be written as the real form: 
\begin{equation}
[-\Delta_0+4 |\tau|^2 (x^2+y^2)] v=2(\lambda+2|\tau|)v. 
\end{equation}

Remember that the real 1-dimensional harmonic oscillator $\hat{H}_0=-\pat^2_x+4|\tau|^2 x^2$ has eigenvalues 
$\lambda_k^0=2|\tau|(1+2k),k=0,1,\cdots$ and the corresponding 1-dimensional eigenspace generated by 
\begin{equation}
\varphi_k^0=H_k(\sqrt{2|\tau|}x)e^{-|\tau|x^2}.
\end{equation}

Now $H_0=-\Delta_0+4|\tau|^2(x^2+y^2)$ is the 2-dimensional harmonic oscillator and there is
$$
H_0=\hat{H}_0(x)+\hat{H}_0(y).
$$
Hence $H_0$ has eigenvalues $\lambda^0_{k,l}=\lambda^0_k+\lambda^0_l, k,l=0,1,\cdots$ and the corresponding eigenforms $\{\varphi^0_k(x)\varphi^0_l(y)\}$. 
So $\Delta_f$ has the spectrum $\lambda_{k,l}^-(|\tau|)=2|\tau|(k+l), k,l=0,1,\cdots$ and the corresponding eigenforms  
\begin{equation}
\varphi^-_{k,l}(z)=(\varphi^0_k(x)\varphi^0_l(y))(-\frac{\tau}{|\tau|}dz+d\zb).
\end{equation}

The equation of $\Delta_f \varphi^+=\lambda \varphi^+$ is equivalent to 
\begin{equation}
H_0 v=2(\lambda-2|\tau|)v,
\end{equation}
where $\varphi^+=v(\frac{\tau}{|\tau|}dz+d\zb)$.

So $\Delta_f|_{E^+}$ has the spectrum $\lambda_{k,l}^+(|\tau|)=2|\tau|(k+l+2), k,l=0,1,\cdots$ and the corresponding eigenforms  
\begin{equation}
\varphi^+_{k,l}(z)=(\varphi^0_k(x)\varphi^0_l(y))(\frac{\tau}{|\tau|}dz+d\zb).
\end{equation}

In particular, we have the ground state for the Schr\'odinger operator.

\begin{prop} $\Delta_f \varphi=0$ has 1-dimensional solution space generated by 
\begin{equation}
\varphi(z)=e^{-\tau |z|^2}( -\frac{\tau}{|\tau|}dz+d\zb)
\end{equation}
\end{prop}

The heat kernel of $\Delta_f$ has the form:
\begin{equation}
K(z,w,t)=\sum_{k,l=0}^\infty e^{-\lambda_{k,l}^- t}\varphi_{k,l}^-(z)\otimes \varphi_{k,l}^-(w)+\sum_{k,l=0}^\infty e^{-\lambda_{k,l}^+t}\varphi_{k,l}^+(z)\otimes \varphi_{k,l}^+(w)
\end{equation}

We have Mehler's formula for real 1-dimensional harmonic oscillators as below:
\begin{align*}
&\sum_{n\ge 0}e^{-(2n+1)t}H_n(x)H_n(y)e^{-\frac{x^2+y^2}{2}}\\
=&(4\pi t)^{-\frac{1}{2}}(\frac{2t}{\sinh 2t})^{\frac{1}{2}}\exp\left( -\frac{1}{4t}\left[ \frac{2t}{\tanh 2t}(x^2+y^2)-\frac{2t}{\sinh2t}(2xy) \right] \right)
\end{align*}

Using the Mehler's formula, we can explicitly compute the heat kernel as follows. Set $z=x+iy, w=u+iv$:
\begin{align*}
&\sum_{k,l=0}^\infty e^{-\lambda_{k,l}^- t}\varphi_{k,l}^-(z)\otimes \varphi_{k,l}^-(w)\\
=&\sum_{k,l=0}^\infty e^{-2|\tau|(k+l)t}(\varphi^0_k(x)\varphi^0_l(y))(-\frac{\tau}{|\tau|}dz+d\zb)\otimes (\varphi^0_k(u)\varphi^0_l(v))(-\frac{\tau}{|\tau|}dw+d\bar{w})\\
=&\left\{\sum_{k=0}^\infty e^{-2|\tau|kt}(\varphi^0_k(x)\varphi^0_k(u)) \right\}\left\{\sum_{l=0}^\infty e^{-2|\tau|lt}(\varphi^0_l(y)\varphi^0_l(v))\right\}(-\frac{\tau}{|\tau|}dz+d\zb)\otimes (-\frac{\tau}{|\tau|}dw+d\bar{w})\\
=&(4\pi |\tau|t)^{-1}(\frac{2|\tau| t}{\sinh 2 |\tau| t})\exp \left\{ {-\frac{1}{2t}[(|z|^2+|w|^2)(\frac{2|\tau| t}{\tanh 2|\tau| t})-(z\bar{w}+\bar{z}w)\frac{2|\tau|t}{\sinh 2|\tau| t}]+2|\tau| t}\right\} \\
&(-\frac{\tau}{|\tau|}dz+d\zb)\otimes (-\frac{\tau}{|\tau|}dw+d\bar{w})\\
=&(4\pi |\tau|t)^{-1}(\frac{2|\tau| t}{\sinh 2 |\tau| t})\exp\left\{{-\frac{|z-w|^2}{2t}\frac{2|\tau| t}{\sinh 2|\tau| t}-|\tau| (|z|^2+|w|^2)\tanh |\tau| t+2|\tau| t}\right\}\\
&(-\frac{\tau}{|\tau|}dz+d\zb)\otimes (-\frac{\tau}{|\tau|}dw+d\bar{w})
\end{align*}

Similarly, we have 
\begin{align*}
&\sum_{k,l=0}^\infty e^{-\lambda_{k,l}^+t}\varphi_{k,l}^+(z)\otimes \varphi_{k,l}^+(w)\\
=&(4\pi |\tau|t)^{-1}(\frac{2|\tau| t}{\sinh 2 |\tau| t})\exp\left\{{-\frac{|z-w|^2}{2t}\frac{2|\tau| t}{\sinh 2|\tau| t}-|\tau| (|z|^2+|w|^2)\tanh |\tau| t-2|\tau| t}\right\}\\
&(\frac{\tau}{|\tau|}dz+d\zb)\otimes (\frac{\tau}{|\tau|}dw+d\bar{w})
\end{align*}

\begin{prop} The heat kernel has the explicit formula:
\begin{align}
&K(z,w,t)=(4\pi |\tau|t)^{-1}(\frac{2|\tau| t}{\sinh 2 |\tau| t})\exp\left\{{-\frac{|z-w|^2}{2t}\frac{2|\tau| t}{\sinh 2|\tau| t}-|\tau| (|z|^2+|w|^2)\tanh |\tau| t+2|\tau| t}\right\}\nonumber\\
&(-\frac{\tau}{|\tau|}dz+d\zb)\otimes (-\frac{\tau}{|\tau|}dw+d\bar{w}) \nonumber\\
+&(4\pi |\tau|t)^{-1}(\frac{2|\tau| t}{\sinh 2 |\tau| t})\exp\left\{{-\frac{|z-w|^2}{2t}\frac{2|\tau| t}{\sinh 2|\tau| t}-|\tau| (|z|^2+|w|^2)\tanh |\tau| t-2|\tau| t}\right\}\nonumber\\
&(\frac{\tau}{|\tau|}dz+d\zb)\otimes (\frac{\tau}{|\tau|}dw+d\bar{w}).
\end{align}
\end{prop}

\subsubsection*{\underline{Heat kernel of $0$ and $2$-forms}}

In this case, the Laplacian operators of $0$ and $2$ forms are equal and have the following form acting on the functions:
$$
\Delta_f=-2\pat_z\pat_\zb+2|\tau|^2|z|^2.
$$
The eigenvalue equation $\Delta_f \varphi=\lambda \varphi$ can be written as 
$$
-\Delta_0 \varphi+4 |\tau|^2 (x^2+y^2)\varphi=2\lambda\varphi. 
$$
Therefore, the eigenvalue $\lambda_{k,l}=2|\tau|(k+l+1),k,l=0,1\cdots$, and the eigenforms are 
\begin{equation}
\varphi_{k,l}(z)=\varphi^0_k(x)\varphi^0_l(y). 
\end{equation}
It is easy to see that the heat kernel is 
\begin{align}
&K(z,w,t)=(4\pi |\tau|t)^{-1}(\frac{2|\tau| t}{\sinh 2 |\tau| t})\exp\left\{{-\frac{|z-w|^2}{2t}\frac{2|\tau| t}{\sinh 2|\tau| t}-|\tau| (|z|^2+|w|^2)\tanh |\tau| t}\right\}.
\end{align}
Furthermore, we have
\begin{align}
&K(z,z,t)=(4\pi |\tau|t)^{-1}(\frac{2|\tau| t}{\sinh 2 |\tau| t})\exp\left\{{-2|\tau| |z|^2\tanh |\tau| t}\right\},
\end{align}
and 
\begin{equation}
\Tr e^{-t\Delta_f}=\Tr K(z,z,t)=(\frac{1}{2\sinh\frac{t}{2}})^2.
\end{equation}

\section{Expansion of heat kernel as $t\to 0$}
 
\subsection{Parametrix of heat equation}

\

Consider the matrix Schr\"odinger operator $\Delta_f=-\Delta+V(\z)\cdot I+B(\z)$ defined on $\C^n$ with coordinates $\z=(z_1,\cdots,z_n), z_j=x_j+iy_j$. Here $V(\z)=|\pat f|^2$, $B(\z)=L_f$ and $\Delta=4\pat_z\pat_\zb=4\sum_i \pat_i\pat_\ib$.  Let $L=\pat_t+\Delta_f$ be the heat operator.

We assume that the parametrix  of $L$ has the following form
\begin{equation}\label{sect-4-equa-0}
P_k(\z,\w,t)=\E_0(\z,\w,t)\E_1(\z,\w,t)\sum_{j=0}^k t^j U_j(\z,\w).
\end{equation}
where the matrix functions $U_j(\z,\w)$ will be determined below.

We have
\begin{align*}
LP_k&=(\pat_t-\Delta_\z+V(\z)\cdot I+B(\z))\left\{\E_0(\z,\w,t)\E_1(\z,\w,t)\sum_{j=0}^k t^j U_j(\z,\w)\right\}\\
&=\sum_{j=0}^k\left\{(\pat_t-\Delta_\z+V(\z)) (\E_0\E_1)t^j U_j+\E_0\E_1 j t^{j-1}U_j-2\nabla_\z \E_0\cdot \nabla_\z U_j \E_1 t^j-2\E_0 \nabla_z \E_1\cdot \nabla_\z U_j t^j\right.\\
&\left. -\E_0\E_1 \Delta_\z U_j t^j+\E_0\E_1B(\z)\cdot U_j t^j\right\}\\
&=\sum_{j=0}^k\E_0\E_1\left\{(t\Delta_\z g-t^2(\nabla_\z g)^2)t^j U_j+[-(\z-\w)\cdot\nabla_\z g-g(\z,\w)+V(\z)]t^jU_j+j t^{j-1}U_j\right.\\
&\left. +(\z-\w)\cdot \nabla_\z U_j t^{j-1}+2\nabla_\z g\cdot \nabla U_j t^{j+1}-t^j\Delta_\z U_j+B(\z)\cdot U_j t^j\right\}\\
&=\E_0\E_1\left\{\sum_{j=0}^{k-1}[(j+1)U_{j+1}+(\z-\w)\cdot \nabla_\z U_{j+1}]t^j+\sum_{j=0}^k(-\Delta_\z U_j+B(\z)\cdot U_j )t^j+\right.\\
&\left.+\sum_{j=0}^k [\Delta_\z g U_j+2\nabla_\z g\cdot \nabla_\z U_j]t^{j+1}+\sum_{j=0}^k [-(\nabla_\z g)^2 U_j]t^{j+2}\right\}\\
&=\E_0\E_1\left\{ \sum_{j=0}^{k-1}[(j+1)U_{j+1}+(\z-\w)\cdot \nabla_\z U_{j+1}-\Delta_\z U_j +B(\z)\cdot U_j]t^j-\Delta_\z U_k t^k+B(\z)\cdot U_k t^k\right.\\
&+\left.\sum_{j=1}^{k+1} [\Delta_\z g U_{j-1}+2\nabla_\z g\cdot \nabla_\z U_{j-1}]t^{j}+\sum_{j=2}^{k+2} [-(\nabla_\z g)^2 U_{j-2}]t^{j}\right\}.
\end{align*}
Set the initial datum
\begin{equation}
U_0\equiv I.
\end{equation}
If we consider the $j=0$ term in the above summation, we obtain the equation of $U_1=U_1(\z,\w)$:
\begin{equation}
U_{1}+(\z-\w)\cdot \nabla_\z U_{1}=-B(\z).
\end{equation}
Denote by $r=|\z-\w|$, then the above equality is equivalent to
$$
\frac{d}{dr}(rU_1)=-B(\z).
$$
$U_1$ can be solved below
\begin{equation}
U_1(\z,\w)=-\frac{1}{r}\int^r_0 B(\tau\cdot \frac{\z-\w}{|\z-\w|}+\w)d\tau=-\int^1_0 B(s(\z-\w)+\w)ds.
\end{equation}
The matrix $U_1(\z,\w)$ is symmetric with respect to the variables $\z$ and $\w$.

For $1\le j\le k-1$, we obtain the identities
\begin{equation}\label{sec1-equ1}
(j+1)U_{j+1}+(\z-\w)\cdot \nabla_\z U_{j+1}-\Delta_\z U_j+B(\z)\cdot U_j+\Delta_\z g U_{j-1}+2\nabla_\z g\cdot \nabla_\z U_{j-1}-(\nabla_\z g)^2 U_{j-2}=0.
\end{equation}
For $j=1$, we have
\begin{equation}
2U_2+(\z-\w)\cdot \nabla_\z U_{2}=\Delta_\z U_1-B(\z)\cdot U_1-\Delta_\z g ,
\end{equation}
or equivalently
\begin{equation}
\frac{d}{dr}(r^2 U_2)=(\Delta_\z U_1-B(\z)\cdot U_1-\Delta_\z g)r
\end{equation}
This gives
\begin{align*}
&U_2(\z,\w)\\
=&\frac{1}{r^2}\int_0^r (\Delta_\z U_1-B\cdot U_1-\Delta_\z g)(\tau\frac{\z-\w}{|\z-\w|}+\w)\tau d\tau\\
=&\int_0^1(\Delta_\z U_1-B\cdot U_1-\Delta_\z g)(\tau(\z-\w)+\w)\tau d\tau\\
=&\int_0^1\tau d\tau \int_0^1(-\Delta_\z V)(\lambda\tau(\z-\w)+\w)\lambda^2 d\lambda+\int^1_0 \tau d\tau
\int^1_0 B(\tau(\z-\w)+\w)\cdot B(s\tau(\z-\w)+\w)ds\\
=&\int_0^1d\lambda \int_0^\lambda (-\Delta_\z B-\Delta_\z V)(\tau(\z-\w)+\w)\tau d\tau+\int^1_0 d\tau\int^\tau_0 B(\tau(\z-\w)+\w)\cdot B(\lambda(\z-\w)+\w)d\lambda\\
=&\int_0^1d\lambda \int_0^\lambda (-\Delta_\z V)(\tau(\z-\w)+\w)\tau d\tau+\int^1_0 d\tau\int^\tau_0 B(\tau(\z-\w)+\w)\cdot B(\lambda(\z-\w)+\w)d\lambda
\end{align*}
Notice that $U_2$ has the form
$$
U_2(\z,\w)=U_2(\z-\w,\w).
$$
For $2\le j\le k-1$, we can solve $U_{j+1}$ from (\ref{sec1-equ1}) in terms of $U_j$ and $U_{j-1}$:
\begin{align*}
U_{j+1}=\frac{1}{r^{j+1}}\int^r_0\left\{\Delta_\z U_j-B\cdot U_j-\Delta_\z g U_{j-1}-2\nabla_\z g\cdot \nabla_\z U_{j-1}+(\nabla_\z g)^2 U_{j-2} \right\}\tau^j d\tau
\end{align*}
Inductively, if $U_j(\z,\w)=U_j(\z-\w,\w)$ and $U_{j-1}(\z,\w)=U_{j-1}(\z-\w,\w)$, $U_{j+1}(\z,\w)$ has the form:
\begin{align*}
&U_{j+1}(\z-\w,\w)\\
=&\int^1_0\left\{\Delta_\z U_j(\tau(\z-\w),\w)-B\cdot U_j(\tau(\z-\w),\w)-\Delta_\z g(\tau(\z-\w),\w) U_{j-1}(\tau(\z-\w),\w)\right.\nonumber\\
-2&\nabla_\z g(\tau(\z-\w),\w)\cdot \nabla_\z U_{j-1}(\tau(\z-\w),\w)+
\left. (\nabla_\z g)^2(\tau(\z-\w),\w) U_{j-2}(\tau(\z-\w),\w) \right\}\tau^j d\tau.
\end{align*}

Up to now, $U_1,\cdots,U_k$ can be uniquely solved if $g(\z,\w)$ and $U_0\equiv 1$ are given above. Finally, we have
\begin{equation}
LP_k=R_k(\z,\w,t),
\end{equation}
where the remainder
\begin{align*}
&R_k(\z,\w,t)\\
=&\E_0\E_1\left\{[-\Delta_\z U_k+B(\z)\cdot U_k+\Delta_\z g U_{k-1}+2\nabla_\z g\cdot \nabla_\z U_{k-1}-(\nabla_\z g)^2 U_{k-2}]
t^k\right.\\
&+[\Delta_\z g U_{k}+2\nabla_\z g\cdot \nabla_\z U_{k}-(\nabla_\z g)^2 U_{k-1}]t^{k+1}+\left.[-(\nabla_\z g)^2 U_{k}]t^{k+2}\right\}\\
=&:\E_0\E_1\left\{\tilde{R}_{k,1}(\z,\w,t)+\tilde{R}_{k,2}(\z,\w,t)+\tilde{R}_{k,3}(\z,\w,t)\right\}\\
=&:\E_0\E_1\tilde{R}_k(\z,\w,t).
\end{align*}
The remainder also has the form $R_k(\z,\w,t)=R_k(\z-\w,\w,t)$.

\subsection{Formal expansion of heat kernel}

\

Suppose that $F$ and $G$ are two smooth matrix valued functions on $\C^n\times \C^n\times [0,\infty]$, then their convolution is defined by
$$
F*G(\z,\w,t)=\int^t_0 \int_{\C^n}F(\z,\x, t-\tau)\cdot G(\x,\w,\tau)d\x d\tau,
$$
if the integration exists.

Now if $R_k$ is proved to be smooth, then it is easy to see that
\begin{equation*}
L(P_k*R_k)=R_k+R_k*R_k.
\end{equation*}
Define
\begin{align}
&R^0_k:=0,R^1_k:=R_k(\z,\w,t),\;P^0_k:=P_k(\z,\w,t),\\
&R^i_k:=\substack{\underbrace{R_k*\cdots*R_k}\\i},\;P^i_k=P^{i-1}_k*R_k=P_k*R^i_k\label{sect-4-equa-61}
\end{align}
Let
\begin{equation}
P(\z,\w,t)=\sum_{i=0}^\infty (-1)^i P_k^i=P_k+P_k*(\sum_{i=0}^\infty (-1)^i R_k^i),
\end{equation}
then we have
$$
LP(\z,\w,t)=0.
$$
In fact, we have
\begin{align*}
LP&=\sum_{i=1}^\infty (-1)^i L(P_k*R_k^i)+R_k=\sum_{i=0}^\infty (-1)^i (R_k^i+R_k*R_k^i)=0.
\end{align*}

We claim that for sufficiently large $k$ $R_k$ is a smooth function and the series $\sum_{i=0}^\infty R^i_k$ converges absolutely and uniformly. We will treat the convergence problem in the next section. 

\subsection{Convergence of the expansion}

\

The following lemmas will be used in the convergence estimate.

\begin{lm}\cite[Chapter IV, Lemma 3]{Cha} Let $x,y,z$ belong to any metric space with metric $d$, then for any $\tau\in (0,t)$, we have
\begin{equation}
\frac{d^2(x,y)}{\tau}+\frac{d^2(y,z)}{t-\tau}\ge \frac{d^2(x,z)}{t}.
\end{equation}
\end{lm}

\begin{lm}\label{sect-4-lemm-4.2} \cite[Page 159]{Cha} $\forall x>0,a>0$, there holds
\begin{equation}
x^a e^{-x}\le a^a e^{-a}.
\end{equation}
\end{lm}

 \begin{lm}\cite[Page 163]{Cha} $\forall \alpha>0, \mu<1$,
 \begin{equation}
 \int^t_0 (t-\tau)^\alpha \tau^{-\mu}d\tau=\frac{\Gamma(1+\alpha)\Gamma(1-\mu)}{\Gamma(2+\alpha-\mu)}t^{\alpha+(1-\mu)},
 \end{equation}
 \end{lm}

\subsubsection*{\underline{Estimate of matrix operator $B$}}

$B$ has the following expression:
$$
B=L_f=-(f_{\mu l}\;\iota_{\pat_\mub}\cdot dz^l\wedge +\overline{f_{\mu l}}\;\iota_{\pat_\mu}\cdot dz^{\lb}\wedge).
$$
Note that
$$
\frac{\pat}{\pat x_i}=\frac{\pat}{\pat z_i}+\frac{\pat }{\pat \zb_i},\;\frac{\pat}{\pat y_i}=\sqrt{-1}(\frac{\pat}{\pat z_i}-\frac{\pat}{\pat \zb_i}).
$$
We have
\begin{align*}
&\frac{\pat}{\pat x_i}B=-(f_{i\mu l}\iota_{\pat_\mub}\cdot dz^l\wedge+\overline{f_{i\mu l}}\iota_{\pat_\mu}\cdot dz^\lb\wedge)\\
&\frac{\pat}{\pat y_i}B=-\sqrt{-1}(f_{i\mu l}\iota_{\pat_\mub}\cdot dz^l\wedge-\overline{f_{i\mu l}}\iota_{\pat_\mu}\cdot dz^\lb\wedge).
\end{align*}
Since the operator $\iota_{\pat_\mub}\cdot dz^l\wedge$ and its complex conjugate have upper bound $1$, we obtain the matrix estimate
\begin{align*}
 ||D^\alpha_\z B||(z_1,\cdots,z_n)&\le 2\sum_{\mu l}|\pat^\alphab_\z f_{\mu l}|(z_1,\cdots,z_n)\quad(=:F^\alpha_B(\z))\\
 &\le 2\sum_{\mu l} t^{\delta_Mq\cdot \alphab-\delta_M(1-q_\mu-q_l)}|\pat^\alphab_\z f_{\mu l}|(z^t_1,\cdots,z^t_n)\\
 &\le t^{\delta_0(|\alpha|+2)-\delta_M}T^{|\alpha|+2} F^\alpha_B(\z_t).
\end{align*}

In the above estimate, we used Lemma \ref{sect-2-Lemm-1}. 
\begin{lm}\label{sec1-lma2}
\begin{equation}
t||D^\alpha_\z B||(\z)\le t F^\alpha_B(\z)\le t^{\delta_0 |\alpha|+\delta_2} T^{|\alpha|+2} F^\alpha_B(\z_t),
\end{equation}
where $\delta_2=2\delta_0+1-\delta_M=\frac{1-2(q_M-q_m)}{2(1-q_M)}$.
\end{lm}

\subsubsection*{\underline{Estimate of $U_i$}}

We estimate the derivatives of $U_1$.
\begin{align*}
||D^\alpha_\z U_1||(\z,\w)&\le ||\int^1_0 D^\alpha_\z B(s(\z-\w)+\w)s^{|\alpha|}ds||\\
&\le \{F^\alpha_B(s(\z-\w)+\w)\}_s\;(=:F^\alpha_1(\z,\w))
\end{align*}

The estimate of $U_2$ is as below. We have
\begin{align*}
&|U_2(\z,\w)|=|U_2(\z-\w,\w)|\\
&\le \{\{(\Delta_\z V)(\lambda\tau(\z-\w),\w)\}_\lambda\}_\tau+\{||B(\tau(\z-\w)+\w)||\cdot\{B(\lambda\tau(\z-\w)+\w)\}_\lambda \}_\tau\\
&\le\{(\Delta_\z V)(\tau(\z-\w)+\w)\}_\tau+\{||B(\tau(\z-\w)+\w)||^2 \}_\tau
\end{align*}
and then
\begin{align*}
&\{U_2(\tau(\z-\w),\w)\}_\tau\le \{|\Delta_\z V|_+(\tau(\z-\w)+\w)\}_\tau+\{||B(\tau(\z-\w)+\w)||^2 \}_\tau\\
&=:F_2(\z,\w)
\end{align*}
Similarly, we have derivative estimates:
\begin{align*}
\{D^\alpha_\z U_2(\z,\w)\}_\tau&:=\{D^\alpha_\z U_2(\tau(\z-\w),\w)\}_\tau\\
&\le \{|D^\alpha_\z \Delta_\z V|_+(\tau(\z-\w),\w)\}_\tau+\sum_{\alpha_1+\alpha_2=\alpha}\frac{\alpha!}{ \alpha_1!\alpha_2!}\{D^{\alpha_1}_\z B\}_\tau\cdot \{D^{\alpha_2}_\z B\}_\tau\\
&=:F^\alpha_2(\z,\w).
\end{align*}

If we set $F^\alpha_0\equiv 1$, We have the simple relations:
\begin{equation}
\{D^\alpha_\z U_i\}_\lambda\le F^\alpha_{i}(\z,\w), \text{for}\; i=0,1,2 \;\text{and any}\;\alpha.
\end{equation}

\begin{lm}\label{sec1-lma3}  Fix $T\ge 1$ and any $\alpha$. In case of $i=0,1,2$,, we have for any $t\in [0,T]$,
\begin{equation}
t^i\{D^\alpha_\z U_i\}_\lambda(\z,\w)\le t^i F^\alpha_{i}(\z,\w)\le t^{\delta_0|\alpha|+i\delta } T^{|\alpha|+4i} F^\alpha_{i}(\z_t,\w_t).
\end{equation}
where $\delta:=\frac{2\delta_3}{3}=\frac{1-3(q_M-q_m)}{3(1-q_M)}$.
\end{lm}

\begin{proof} Since $U_0\equiv I$, $i=0$ case is trivial. For $i=1$ case, we have

\begin{align*}
&t F^\alpha_{1}(\z,\w)=t\{F^\alpha_B(s(\z-\w),\w)\}_s\\
&\le t^{\delta_0|\alpha|+\delta_2}T^{|\alpha|+2}F^\alpha_1(\z_t,\w_t)\le t^{\delta_0|\alpha|+\delta}T^{|\alpha|+3}F^\alpha_1(\z_t,\w_t)
\end{align*}

For $i=2$ case, we have
\begin{align*}
&t^2 F^\alpha_{2}(\z,\w)\\
=&t^2 \{|D^\alpha_\z\Delta_\z V|_+ (\tau(\z-\w)+\w)\}_\tau+t^2\sum_{\alpha_1+\alpha_2=\alpha}\frac{\alpha!}{\alpha_1!\alpha_2!}\{D^{\alpha_1}_\z B\}_\tau \cdot \{D^{\alpha_2}B\}_\tau(\z,\w)\\
\le&  t^{\delta_0(|\alpha|+4)-2\delta_M+2}T^{|\alpha|+4} \{|D^\alpha_\z\Delta_\z V|_+ \}_\tau(\z_t,\w_t)\\
&+\sum_{\alpha_1+\alpha_2=\alpha}\frac{\alpha!}{\alpha_1!\alpha_2!}t^{\delta_0|\alpha_1|+\delta_2} t^{\delta_0|\alpha_2|+\delta_2} T^{|\alpha_1|+|\alpha_2|+4}\{D^{\alpha_1}_\z B\}_\tau \cdot \{D^{\alpha_2}B\}_\tau\\
\le& t^{\delta_0|\alpha|+2\delta_2}T_1^{|\alpha|+4}F^\alpha_{2}(\z_t,\w_t)\\
\le &t^{\delta_0|\alpha|+2\delta}T^{|\alpha|+6} F^\alpha_{2}(\z_t,\w_t)
\end{align*}
This proves $i=2$ case.
\end{proof}

Now we inductively define $F^\alpha_{j+1}$ for $j\ge 2$ and any $\alpha$ in terms of $F^\beta_{j},F^\beta_{j-1}$ and $F^\beta_{j-2}$ for $|\beta|\le |\alpha|$.

We have
\begin{align*}
&\{D^\alpha_\z U_{j+1}\}_\lambda\\
\le& \{D^\alpha_\z\Delta_\z U_j\}_\lambda+\sum_{\alpha_1+\alpha_2=\alpha}\frac{\alpha!}{\alpha_1!\alpha_2!}\{D^{\alpha_1} B\cdot D^{\alpha_2} U_{j}\}_\lambda+\sum_{\alpha_1+\alpha_2=\alpha}\frac{\alpha!}{\alpha_1!\alpha_2!}\{D^{\alpha_1}\Delta_\z g\cdot D^{\alpha_2} U_{j-1}\}_\lambda\\
&+2\sum_{\alpha_1+\alpha_2=\alpha}\frac{\alpha!}{\alpha_1!\alpha_2!}\{D^{\alpha_1}\nabla_\z g\cdot D^{\alpha_2}\nabla_\z U_{j-1}\}_\lambda+\sum_{\alpha_1+\alpha_2+\alpha_3=\alpha}\frac{\alpha!}{\alpha_1!\alpha_2!\alpha_3!}\{D^{\alpha_1}\nabla_\z g\cdot D^{\alpha_2}_\z \nabla_\z g D^{\alpha_3}U_{j-2}\}_\lambda\\
\le &\sum_{k=1}^n F^{\alpha+2e_k}_j+\sum_{\alpha_1+\alpha_2=\alpha}\frac{\alpha!}{\alpha_1!\alpha_2!}F^{\alpha_1}_B F^{\alpha_2}_j+\sum_{\alpha_1+\alpha_2=\alpha}\frac{\alpha!}{\alpha_1!\alpha_2!}\{|D^{\alpha_1}\Delta_\z V|_+\}_\lambda\cdot F^{\alpha_2}_{j-1}\\
&+2\sum_{k=1}^n\sum_{\alpha_1+\alpha_2=\alpha}\frac{\alpha!}{\alpha_1!\alpha_2!}\{|D^{\alpha_1+e_k}
V|_+\}_\lambda F^{\alpha_2+e_k}_{j-1}\\
&+\sum_{k=1}^n\sum_{\alpha_1+\alpha_2+\alpha_3=\alpha}\frac{\alpha!}{\alpha_1!\alpha_2!\alpha_3!}\{|D^{\alpha_1+e_k}
V|_+\}_\lambda\cdot \{|D^{\alpha_2+e_k}_\z  V|_+ \}_\lambda
F^{\alpha_3}_{j-2},
\end{align*}
where $e_k$ are the standard unit $n$ dimension vectors. We denote the right hand side of the above inequality by $F^\alpha_{j+1}$. It is obvious that each $F^\alpha_j$ has the form
$$
F^\alpha_j(\z,\w)=F^\alpha_j(\z-\w,\w).
$$
\begin{lm}\label{sect-4-lemm-4.6} For any $i\ge 1$, $T\ge 1$, any $\alpha$ and any $t\in [0,T]$, there is
\begin{equation}\label{sec1-equ2}
t^i\{D^\alpha_\z U_i\}_\lambda(\z,\w)\le t^i F^\alpha_{i}(\z,\w)\le t^{\delta_0|\alpha|+i\delta }T^{|\alpha|+4i}F^\alpha_{i}(\z_t,\w_t),
\end{equation}

\end{lm}

\begin{proof} Prove by induction. The inequality (\ref{sec1-equ2}) holds for $i=0,1,2$ due to Lemma \ref{sec1-lma3} .
Now we assume it holds for $i=j-2,j-1,j$, and we want to prove that it holds for $i=j+1$. By Lemma \ref{sec1-lma1} and \ref{sec1-lma2}, we have
\begin{align*}
&t^{j+1}F^{\alpha}_{j+1}(\z,\w)\\
&=\sum_{k=1}^n t^{j+1}F^{\alpha+2e_k}_j(\z,\w)+\sum_{\alpha_1+\alpha_2=\alpha}\frac{\alpha!}{\alpha_1!\alpha_2!}t F^{\alpha_1}_B t^{j}F^{\alpha_2}_j(\z,\w)\\
&+\sum_{\alpha_1+\alpha_2}\frac{\alpha!}{\alpha_1!\alpha_2!}t^2\{|D^{\alpha_1}\Delta_\z V|_+\}_\lambda(\z,\w)\cdot t^{j-1}F^{\alpha_2}_{j-1}(\z,\w)\\
&+2\sum_{k=1}^n\sum_{\alpha_1+\alpha_2=\alpha}\frac{\alpha!}{\alpha_1!\alpha_2!}t^2\{|D^{\alpha_1+e_k} V|_+\}_\lambda(\z,\w) t^{j-1}F^{\alpha_2+e_k}_{j-1}(\z,\w)\\
&+\sum_{k=1}^n\sum_{\alpha_1+\alpha_2+\alpha_3=\alpha}\frac{\alpha!}{\alpha_1!\alpha_2!\alpha_3!}t^3\{|D^{\alpha_1+e_k} V|_+\}_\lambda(\z,\w)\cdot \{|D^{\alpha_2+e_k}_\z  V|_+ \}_\lambda(\z,\w) t^{j-2}F^{\alpha_3}_{j-2}(\z,\w)\\
\le &\sum_{k=1}^n t^{\delta_0 |\alpha|+\delta (j+1)} T^{|\alpha|+4+4j} F^{\alpha+2e_k}_j(\z_t,\w_t)+\sum_{\alpha_1+\alpha_2=\alpha}\frac{\alpha!}{\alpha_1!\alpha_2!}t^{\delta_0|\alpha|+(j+1)\delta}T^{|\alpha|+4+4j} F^{\alpha_1}_B F^{\alpha_2}_j(\z_t,\w_t)\\
&+\sum_{\alpha_1+\alpha_2}\frac{\alpha!}{\alpha_1!\alpha_2!}t^{\delta_0|\alpha|+\delta(j+1)}T^{|\alpha|+4+4j}\{|D^{\alpha_1}\Delta_\z V|_+\}_\lambda(\z_t,\w_t)\cdot F^{\alpha_2}_{j-1}(\z_t,\w_t)\\
&+2\sum_{k=1}^n\sum_{\alpha_1+\alpha_2=\alpha}\frac{\alpha!}{\alpha_1!\alpha_2!}t^{\delta_0|\alpha|+\delta(j+1)}T^{|\alpha|+4+4j}\{|D^{\alpha_1+e_k} V|_+\}_\lambda(\z_t,\w_t) F^{\alpha_2+e_k}_{j-1}(\z_t,\w_t)\\
+&\sum_{k=1}^n\sum_{\alpha_1+\alpha_2+\alpha_3=\alpha}\frac{\alpha!}{\alpha_1!\alpha_2!\alpha_3!}t^{\delta_0|\alpha|+\delta(j+1)}T^{|\alpha|+4(j+1)}\{|D^{\alpha_1+e_k} V|_+\}_\lambda(\z_t,\w_t)\\
&\cdot \{|D^{\alpha_2+e_k}_\z  V|_+ \}_\lambda(\z_t,\w_t)F^{\alpha_3}_{j-2}(\z_t,\w_t)\\
&\le  t^{\delta_0|\alpha|+\delta(j+1)}T^{|\alpha|+4(j+1)}F^{\alpha}_{j+1}(\z_t,\w_t).
\end{align*}
This proves the conclusion.
\end{proof}

\subsubsection*{\underline{Estimate of the remainder}}

Now we have the estimate of the remainder $\tilde{R}_k(\z,\w,t)$.
\begin{align*}
&|\tilde{R}_k(\z,\w,t)|\\
\le& \sum^n_{l=1} t^k F^{2e_l}_k+t^k F_B F_k+t^k\{|\Delta_z V|_+\}_\lambda F^0_{k-1}+2\sum_l t^k\{|D^{e_l}_\z V|_+\}_\lambda F^{e_l}_{k-1}+t^k\{(|\nabla_z V|_+)^2\}_\lambda F^0_{k-2}\\
&+t^{k+1}\{|\Delta_\z V|_+\}_\lambda F^0_k+2t^{k+1}\sum_l \{|\nabla^{e_l}_\z V|_+\}_\lambda F^{e_l}_k+t^{k+1}\{(|\nabla_\z V|_+)^2\}_\lambda F^0_{k-1}\\
&+t^{k+2}\{(|\nabla_\z V|_+)^2\}_\lambda F^0_k\;\;(\text{defined as}\;\tilde{F}_k(\z,\w,t))
\end{align*}

\begin{align*}
&\le \sum^n_{l=1} t^{2\delta_0+\delta k} T^{2+4k} F^{2e_l}_k(\z_t,\w_t)+t^{\delta_2-1+\delta k}T^{4k+2}F_BF_k(\z_t,\w_t)\\
&+t^{-1+2\delta_2-\delta+\delta k} T^{4k}\{|\Delta_z V|_+\}_\lambda(\z_t,\w_t) F^0_{k-1}(\z_t,\w_t)+2\sum_l t^{-1+2\delta_2+\delta (k-1)} T^{4k}\{|D^{e_l}_\z V|_+\}_\lambda(\z_t,\w_t) F^{e_l}_{k-1}(\z_t,\w_t)\\
&+t^{-1+2\delta_3-2\delta+\delta k} T^{4k-2}\{(|\nabla_z V|_+)^2\}_\lambda(\z_t,\w_t) F^0_{k-2}(\z_t,\w_t)+t^{-1+2\delta_2+\delta k} T^{4(k+1)}\{|\Delta_\z V|_+\}_\lambda(\z_t,\w_t) F^0_k(\z_t,\w_t)\\
&+2t^{-\frac{1}{2}+\delta_3+\delta_0+\delta k} T^{4(k+1)}\sum_l \{|\nabla^{e_l}_\z V|_+\}_\lambda(\z_t,\w_t) F^{e_l}_k(\z_t,\w_t)\\
&+t^{-1+2\delta_3-\delta+\delta k} T^{4k+2}\{(|\nabla_\z V|_+)^2\}_\lambda(\z_t,\w_t) F^0_{k-1}(\z_t,\w_t)+t^{-1+2\delta_3+\delta k} T^{4k+6}\{(|\nabla_\z V|_+)^2\}_\lambda(\z_t,\w_t) F^0_k(\z_t,\w_t)\\
\le& t^{-1+\delta k} T^{4k+10}\tilde{F}_k(\z_t,\w_t,1).
\end{align*}

The estimate of higher derivatives follows:

\begin{align*}
&|D^\alpha_z\tilde{R}_k(\z,\w,t)|\\
&\le \sum^n_{l=1} t^k F^{\alpha+2e_l}_k+t^k\sum_{\alpha_1+\alpha_2=\alpha} \frac{\alpha!}{\alpha_1!\alpha_2 !}F^{\alpha_1}_B F^{\alpha_2}_k+t^k\sum_{\alpha_1+\alpha_2=\alpha} \frac{\alpha!}{\alpha_1!\alpha_2 !}\{|D^{\alpha_1}\Delta_z V|_+\}_\lambda F^{\alpha_2}_{k-1}\\
&+2\sum_l \sum_{\alpha_1+\alpha_2=\alpha} t^k\frac{\alpha!}{\alpha_1!\alpha_2 !}\{|D^{\alpha_1+e_l}_\z V|_+\}_\lambda F^{\alpha_2+e_l}_{k-1}+t^k\sum_{\alpha_1+\alpha_2+\alpha_3=\alpha}\frac{\alpha!}{\alpha_1!\alpha_2 !\alpha_3 !}\{|D^{\alpha_1}\nabla_z V|_+\}_\lambda\cdot \{|D^{\alpha_2}\nabla_z V|_+\}_\lambda F^{\alpha_3}_{k-2}\\
&+t^{k+1}\sum_{\alpha_1+\alpha_2=\alpha} \frac{\alpha!}{\alpha_1!\alpha_2 !}\{|D^{\alpha_1}\Delta_\z V|_+\}_\lambda F^{\alpha_2}_k+2t^{k+1}\sum_l \sum_{\alpha_1+\alpha_2=\alpha} \frac{\alpha!}{\alpha_1!\alpha_2 !}\{|D^{\alpha_1}\nabla^{e_l}_\z V|_+\}_\lambda F^{\alpha_2+e_l}_k\\
&+t^{k+1}\sum_{\alpha_1+\alpha_2+\alpha_3=\alpha}\frac{\alpha!}{\alpha_1!\alpha_2 !\alpha_3 !}\{|D^{\alpha_1}\nabla_z V|_+\}_\lambda\cdot \{|D^{\alpha_2}\nabla_z V|_+\}_\lambda F^{\alpha_3}_{k-1}\\
&+t^{k+2}\sum_{\alpha_1+\alpha_2+\alpha_3=\alpha}\frac{\alpha!}{\alpha_1!\alpha_2 !\alpha_3 !}\{|D^{\alpha_1}\nabla_z V|_+\}_\lambda\cdot \{ |D^{\alpha_2}\nabla_z V|_+\}_\lambda F^{\alpha_3}_{k}\;\;(\text{defined as}\; \tilde{F}_k^{\alpha}(\z,\w,t))
\\
&\le \sum^n_{l=1} t^{2\delta_0+\delta_0 |\alpha|+\delta k} T^{2+|\alpha|+4k}F^{\alpha+2e_l}_k(\z_t,\w_t)+t^{\delta_2-1+\delta_0 |\alpha|+\delta k}T^{4k+|\alpha|+2}\sum_{\alpha_1+\alpha_2=\alpha} \frac{\alpha!}{\alpha_1!\alpha_2 !}F^{\alpha_1}_B F^{\alpha_2}_k\\
&+t^{-1+2\delta_2-\delta+\delta_0 |\alpha|+\delta k} T^{|\alpha|+4k}\sum_{\alpha_1+\alpha_2=\alpha} \frac{\alpha!}{\alpha_1!\alpha_2 !}\{|D^{\alpha_1}\Delta_z V|_+\}_\lambda F^{\alpha_2}_{k-1}(\z_t,\w_t)\\
&+2\sum_l \sum_{\alpha_1+\alpha_2=\alpha}t^{-1+\delta_2+\delta_0 |\alpha|+\delta k} T^{|\alpha|+4k+1}\frac{\alpha!}{\alpha_1!\alpha_2 !}\{|D^{\alpha_1+e_l}_\z V|_+\}_\lambda F^{\alpha_2+e_l}_{k-1}(\z_t,\w_t)\\
&+t^{-1+2\delta_3-2\delta+\delta_0 |\alpha|+\delta k} T^{|\alpha|+4k-2} \sum_{\alpha_1+\alpha_2+\alpha_3=\alpha}\frac{\alpha!}{\alpha_1!\alpha_2 !\alpha_3 !}\{|D^{\alpha_1}\nabla_z V|_+\}_\lambda\cdot\{| D^{\alpha_2}\nabla_z V|_+\}_\lambda F^{\alpha_3}_{k-2}(\z_t,\w_t)
\end{align*}

\begin{align*}
&+t^{-1+2\delta_2+\delta_0 |\alpha|+\delta k} T^{|\alpha|+4(k+1)}\sum_{\alpha_1+\alpha_2=\alpha} \frac{\alpha!}{\alpha_1!\alpha_2 !}\{|D^{\alpha_1}\Delta_\z V|_+\}_\lambda F^{\alpha_2}_k(\z_t,\w_t)\\
&+2t^{-1+2\delta_2+\delta_0 |\alpha|+\delta k} T^{|\alpha|+4(k+1)}\sum_l \sum_{\alpha_1+\alpha_2=\alpha} \frac{\alpha!}{\alpha_1!\alpha_2 !}\{|D^{\alpha_1}\nabla^{e_l}_\z V|_+\}_\lambda F^{\alpha_2+e_l}_k(\z_t,\w_t)
\\
&+t^{-1+2\delta_3-\delta+\delta_0 |\alpha|+\delta k} T^{|\alpha|+4k+2} \sum_{\alpha_1+\alpha_2+\alpha_3=\alpha}\frac{\alpha!}{\alpha_1!\alpha_2 !\alpha_3 !}\{|D^{\alpha_1}\nabla_z V|_+\}_\lambda\cdot \{|D^{\alpha_2}\nabla_z V|_+\}_\lambda F^{\alpha_3}_{k-1}(\z_t,\w_t)\\
&+t^{-1+2\delta_3+\delta_0 |\alpha|+\delta k} T^{|\alpha|+4k+6}\sum_{\alpha_1+\alpha_2+\alpha_3=\alpha}\frac{\alpha!}{\alpha_1!\alpha_2 !\alpha_3 !}\{|D^{\alpha_1}\nabla_z V|_+\}\cdot\{| D^{\alpha_2}\nabla_z V|_+\}_\lambda F^{\alpha_3}_{k}(\z_t,\w_t)\\
&\le  t^{-1+\delta k+\delta_0 |\alpha|} T^{4k+10+|\alpha|}\tilde{F}_k^\alpha(\z_t,\w_t,1).
\end{align*}

\begin{lm}\label{sect4-lemm-4.7} For any $\alpha$, there holds
\begin{equation*}
(1)\quad |D^\alpha_z\tilde{R}_k(\z,\w,t)|\le \tilde{F}_k^\alpha(\z,\w,t)\le  t^{-1+\delta k+\delta_0 |\alpha|} T_1^{4k+10+|\alpha|}\tilde{F}_k^\alpha(\z_t,\w_t,1)
\end{equation*}
For any $0<l\le k$, there holds
\begin{align*}
(2)\quad |D^l_t D^\alpha_\z\tilde{R}_k(\z,\w,t)|&\le \frac{(k+2)!}{(k+2-l)!}t^{-l}\tilde{F}^\alpha_k(\z,\w,t)\nonumber\\
&\le \frac{(k+2)!}{ (k+2-l)!}t^{-l-1+\delta k+\delta_0 |\alpha|} T_1^{4k+10+|\alpha|}\tilde{F}_k^\alpha(\z_t,\w_t,1)
\end{align*}
\end{lm}

\begin{prop}\label{sect-4-prop-1} For any $l\in \N$ and $\alpha$, we have the estimate of the remainder $R_k(\z,\w,t)$:
\begin{align}
& |D^\alpha_\z R_k(\z,\w,t)|\le t^{\delta k-\frac{3}{2}|\alpha|-1}T^{4k+4|\alpha|+10}\E_0(\z,\w,t) \{\E_1(\z_t,\w_t,1)\nonumber\\
&\sum_{\alpha_1+\alpha_2+\alpha_3=\alpha}\frac{\alpha}{\alpha_1!\alpha_2!\alpha_3!} |P^{|\alpha_1|}_0|_+(\z_t-\w_t,1)\cdot F_g^{\alpha_2}(\z_t,\w_t,1)\cdot \tilde{F}^{\alpha_3}_k(\z_t,\w_t,1)\}\label{sec1-equa-3}\\
&|D^l_t R_k(\z,\w,t)|\le t^{\delta k-3l-1}T^{4k+10+2l}\E_0(\z,\w,t)\big\{\E_1(\z_t,\w_t,1)\nonumber\\
&\sum_{l_1+l_2+l_3=l}\frac{l!}{ l_1! l_2! l_3!}\frac{(k+2)!}{ (k+2-l_3)!}\cdot |P^{2l_1}_1|_+(|\z_t-\w_t|,1)\cdot |g(\z_t,\w_t,1)|^{l_2}\cdot\tilde{F}_k(\z_t,\w_t,1)\big\}\label{sec1-equa-4}
\end{align}
Furthermore, if we denote the terms inside $\{\cdot \}$ of (\ref{sec1-equa-3}) and (\ref{sec1-equa-4}) by $\hat{F}^\alpha_{k,1}(\z_t,\w_t)$ and $\hat{F}^l_{k,2}(\z_t,\w_t)$ respectively, then there exists constants $c^{\alpha}_{k,1}, c^{l}_{k,2}, $ depending only on $k,\alpha, l$ and $V$ such that
\begin{align*}
(1)\quad&\hat{F}^\alpha_{k,i}(\z_t,\w_t),\;\int_{\C^n}\hat{F}^\alpha_{k,1}(\z_t,\w_t)d\z_t , \;\int_{\C^n}\hat{F}^\alpha_{k,1}(\z_t,\w_t)d\w_t\le  c^{\alpha}_{k,1}\\
(2)\quad&\hat{F}^l_{k,2}(\z_t,\w_t),\;\int_{\C^n}\hat{F}^l_{k,2}(\z_t,\w_t)d\z_t , \;\int_{\C^n}\hat{F}^l_{k,2}(\z_t,\w_t)d\w_t\le  c^l_{k,2}.
\end{align*}
\end{prop}

\begin{proof} By Lemma \ref{sect-2-lemm-2.13}, \ref{sect-2-lemm-2.14} and \ref{sect4-lemm-4.7}, we have
\begin{align*}
&|D^\alpha_\z R_k(\z,\w,t)|\\
&\le\sum_{\alpha_1+\alpha_2+\alpha_3=\alpha}\frac{\alpha}{\alpha_1!\alpha_2!\alpha_3!} |D^{\alpha_1}\E_0 D^{\alpha_2}\E_1 D^{\alpha_3}\tilde{R}_k|\\
&\le \sum_{\alpha_1+\alpha_2+\alpha_3=\alpha}\frac{\alpha}{\alpha_1!\alpha_2!\alpha_3!} t^{-|\alpha_1|}|P^{|\alpha_1|}_0(\z-\w,t)|_+\cdot |G^{\alpha_2}(\z,\w,t)|\cdot \tilde{F}^{\alpha_3}_k(\z,\w,t)\E_0\E_1\\
&\le \sum_{\alpha_1+\alpha_2+\alpha_3=\alpha}\frac{\alpha}{\alpha_1!\alpha_2!\alpha_3!} t^{-(1+\delta_Mq_M)|\alpha_1|+(-\frac{1}{2}+\delta_3)|\alpha_2|-1+\delta k+\delta_0|\alpha_3|}T^{4k+10+|\alpha_1|+3|\alpha_2|+|\alpha_3|}\\
&\cdot |P^{|\alpha_1|}_0|_+(\z_t-\w_t,1)\cdot F^{\alpha_2}_g(\z_t,\w_t,1)\cdot \tilde{F}^{\alpha_3}_k(\z_t,\w_t,1)\E_1(\z_t,\w_t,1)\E_0(\z,\w,t).\\
&\le t^{\delta k-\frac{3}{2}|\alpha|-1}T^{4k+4|\alpha|+10}\sum_{\alpha_1+\alpha_2+\alpha_3=\alpha}\frac{\alpha}{\alpha_1!\alpha_2!\alpha_3!} |P^{|\alpha_1|}_0|_+(\z_t-\w_t,1)\cdot F^{\alpha_2}_g(\z_t,\w_t,1)\\
&\cdot \tilde{F}^{\alpha_3}_k(\z_t,\w_t,1)\E_1(\z_t,\w_t,1)\E_0(\z,\w,t).
\end{align*}
For any $l\in \N$, there is
\begin{align*}
&|D^l_t R_k(\z,\w,t)|\\
&\le\sum_{l_1+l_2+l_3=l}\frac{l!}{ l_1! l_2! l_3!}|D^{l_1}_t\E_0 D^{l_2}_t\E_1 D^{l_3}_t\tilde{R}_k|\\
&\le \sum_{l_1+l_2+l_3=l}\frac{l!(k+2)!}{ l_1! l_2! l_3!(k+2-l_3)!}t^{-2l_1}|P^{2l_1}_1(|\z-\w|,t)|\cdot\E_0(\z,\w,t)\cdot |(-g)^{l_2}|\E_1(\z,\w,t)
t^{-l_3}\tilde{F}_k(\z,\w,t)\\
&\le \sum_{l_1+l_2+l_3=l}\frac{l!(k+2)!}{ l_1! l_2! l_3!(k+2-l_3)!} t^{-2(1+\delta_Mq_M)l_1-2\delta_M(1-q_m)l_2-l_3-1+\delta k}T_1^{2l_1+2l_2+4k+10}\\
&\cdot |P^{2l_1}_1(|\z_t-\w_t|,1)|_+\cdot |g(\z_t,\w_t,1)|^{l_2}\cdot\E_1(\z_t,\w_t,1)\cdot\tilde{F}_k(\z_t,\w_t,1)\E_0(\z,\w,t)\\
&\le t^{\delta(k-2)-3l-1}T^{4k+10+2l}\sum_{l_1+l_2+l_3=l}\frac{l!(k+2)!}{ l_1! l_2! l_3!(k+2-l_3)!}\cdot |P^{2l_1}_1(|\z_t-\w_t|,1)|_+\\
&\cdot |g(\z_t,\w_t,1)|^{l_2}\cdot\E_1(\z_t,\w_t,1)\cdot\tilde{F}_k(\z_t,\w_t,1)\E_0(\z,\w,t)
\end{align*}
Note that $\hat{F}^\alpha_{k,1}$ and $\hat{F}^l_{k,2}$ are the finite summation of those terms having the following form with $\z_t,\w_t$ replaced by $\z,\w$:
$$
|\z-\w|^{s_1}\prod_{l=1}^{s_2}\{D^{\alpha_l}_\z g(\z,\w, 1)\}_\lambda e^{-g(\z,\w)}.
$$
An easy computation shows that 
$$
g(\z,\w)=\int^1_0|\tau(\z-\w)+\w|^2d\tau=\frac{1}{3}|\z-\w|^2+(\z-\w)\cdot \w+|\w|^2\ge \frac{1}{4}|\w|^2.
$$
Since $g(\z,\w)$ is symmetric about $\z$ and $\w$, in fact we have
\begin{equation}\label{sect-4-equa-5}
g(\z,\w)\ge \frac{1}{8}(|\z|^2+|\w|^2).
\end{equation}
By inequality (\ref{sec0-esti2}), (\ref{sec0-esti3}), (\ref{sect-4-equa-5}) and Corollary \ref{sect-2-lemm-2.9}, we have
\begin{align*}
&|\z-\w|^{s_1}\prod_{l=1}^{s_2}\{D^{\alpha_l}_\z g(\z,\w, 1)\}_\lambda e^{-g(\z,\w)}\\
&\le |\z-\w|^{s_1}\prod_{l=1}^{s_2}\max_{\x\in \overline{\z\w}}|D^{\alpha_l}_\z V(\x)|_+ e^{-g(\z,\w)}\\
&\le |\z-\w|^{s_1}\prod_{l=1}^{s_2}\max_{\x\in \overline{\z\w}}C_{\alpha_l}( V(\x)+1)^{2-|\alpha_l|q_m-2q_m} e^{-g(\z,\w)}\\
&\le \left(2^{s_1-1}\prod_{l=1}^{s_2}C_{\alpha_l} e \right)(|\z|^{s_1}+|\w|^{s_1}) \max_{\x\in \overline{\z\w}}( V(\x)+1)^{(2-2q_m)s_2} e^{-\frac{1}{c}\int^1_0 |\tau(\z-\w)+\w|^2}\\
&\le \left( 2^{s_1-1}\prod_{l=1}^{s_2}C_{\alpha_l} e \right)(|\z|^{s_1}+|\w|^{s_1})(C|\z|^{2\max_l |b_l|}+C|\w|^{2\max_l |b_l|}+2)^{(2-2q_m)s_2} e^{-\frac{1}{8c}(|\z|^2+|\w|^2)}\\
&\le c(k.\alpha)(|z|^2+|\w|^2+1)^\beta e^{-\frac{1}{8c}(|\z|^2+|\w|^2+1)}\\
&\le c(k,\alpha)\beta^\beta e^{-\beta}=:c^\alpha_{k,1},
\end{align*}
where $\beta=\frac{s_1}{2}+4\max_l |b_l|\{(1-q_m)s_2$ and we used Lemma \ref{sect-4-lemm-4.2} in the last second inequality.

Hence we get a uniform estimate
\begin{equation}
\hat{F}^\alpha_{k,1}\le c^{\alpha}_{k,1},
\end{equation}
where the constant $c^{\alpha}_{k,1}$ only depends on $k,\alpha$ and the constants appearing in the polynomial $f$.

Now the integral estimate follows obviously from the above upper bound estimate. Since the variable $\z$ and $\w$ have symmetry, we can get the similar integrability estimate w. r. t. $\w$.  So eventually we obtain the estimate about $\hat{F}^\alpha_{k,1}$ and $\hat{F}^l_{k,2}$.
\end{proof}

\begin{crl} Assume that $q_M-q_m<\frac{1}{3}$ and let $\delta=\frac{1-3(q_M-q_m)}{3(1-q_M)}$. Given $l_0\in \N$. If $k\in \N$ satisfies $ k>\frac{3l_0+n+1}{\delta}$, then $\tilde{R}$ is a continuous function defined on $\C^n\times \C^n\times [0,1]$ and has up to $2l_0$-th derivatives w. r. t. $\z$ and up to $l_0$-th derivatives w. r. t. $t$. It also holds for any $\alpha,|\alpha|\le 2l_0$,
$$
\lim_{t\to 0}||D^\alpha_\z R_k(\cdot,\cdot, t)||_{C(\C^{n})}=0,
$$
and for any $0\le a\le l_0$,
$$
\lim_{t\to 0}||D^a_t R_k(\cdot,\cdot, t)||_{C(\C^n)}=0.
$$
\end{crl}

\begin{crl}Assume that $q_M-q_m<\frac{1}{3}$ and let $k\in \N$ satisfies $ k>\frac{3l_0+n+1}{\delta}$. Then for any such $k$, there exists a family of smooth kernel $P_k(\z,\w,t)$ such that, for every $l_0\in \N$
\begin{enumerate}
\item for every $T>0$, the operator $P_k(\cdot,\cdot,t)$ form a uniformly bounded family of operators on the space $C^{2l_0}_0(\C^n)$ for every $t\in (0,T]$;
\item for every $h(\z)\in C^{2l_0}_0(\C^n)$, we have
$$
\lim_{t\to 0}P_k*h=h
$$
with respect to the norm $||\cdot||_{C^{2l_0}(\C^n)}$.
\item the remainder $R_k(\z,\w,t)$ satisfies the estimate
$$
||R_k(\cdot,\cdot)||_{C^{2l_0}}\le Ct^{\delta k-3l_0-n-1}.
$$
\end{enumerate}
\end{crl}

\begin{proof} It suffices to prove $\lim_{t\to 0}P_k*f=f$ in the continuous norm. This is obtained by the dominant convergence theorem.
\end{proof}

\begin{lm}Let $a>-1,b>-1$, then we have
\begin{equation}
\int^t_0(t-\tau)^a\tau^b d\tau\le \frac{t^{a+b+1}}{\max(a,b)+1}.
\end{equation}
\end{lm}

\begin{proof}We have
\begin{align*}
&\int^t_0(t-\tau)^a\tau^b d\tau=t^{a+b+1}\int^1_0(1-\tau)^a\tau^b d\tau\\
\le& t^{a+b+1}\min\{\frac{1}{a+1},\frac{1}{b+1}\}=\frac{t^{a+b+1}}{\max(a,b)+1}
\end{align*}
\end{proof}

Consider the estimate of the convolution $R^2_k(\z,\w,t)=R_k*R_k$ for $k$ large. We have
\begin{align*}
&|D^\alpha_z R^2_k(\z,\w,t)|\\
\le& \int^t_0\int_{\C^n} |D^\alpha_\z R_k(\z,\x,t-\tau)|\cdot |R_k(\x,\w,\tau)|d\x d\tau\\
\le& \int^t_0 (4\pi)^{-2n}e^{-\frac{|\z-\w|^2}{4t}}(t-\tau)^{-n}\tau^{-n}(t-\tau)^{\delta k-\frac{3|\alpha|}{2}-1}T^{4k+4|\alpha|+10}\sup_{\x_t\in \C^n}\hat{F}^\alpha_{k,1}(\z_{t-\tau},\x_{t-\tau})\\
&\tau^{\delta k-1}T^{4k+10}\tau^{-2\sum_i q_i}\int_{\C^n} \hat{F}_{k,1}(\x_\tau,\w_\tau)d\x_\tau \\
\le& \int^t_0 (t-\tau)^{\delta k-\frac{3|\alpha|}{2}-n-1}\tau^{\delta(k-2)-1-n-2\sum_i q_i} d\tau [c^\alpha_{k,1}T^{4k+4|\alpha|+10}][c_{k,1}T^{4k+10}]e^{-\frac{|\z-\w|^2}{4t}}\\
\le& \frac{t^{2\delta k-2n-1-2\sum_i q_i-\frac{3}{2}|\alpha|}}{(\delta k-n-2\sum_i q_i)-\frac{3}{2}|\alpha|}[c^\alpha_{k,1}T^{4k+4|\alpha|+10}][c_{k,1}T^{4k+10}] e^{-\frac{|\z-\w|^2}{4t}}
\end{align*}

In particular, we have
\begin{equation}
|R^2_k(\z,\w,t)|\le \frac{t^{2[\delta k-n]-2\sum_i q_i-1}}{\delta k-n-2\sum_i q_i}[c_{k,1}T_1^{4k+10}]^2e^{-\frac{|\z-\w|^2}{4t}}.
\end{equation}

Suppose that we have the estimate
\begin{align}
&|R^l_k(\z,\w,t)|\nonumber\\
\le &\frac{t^{l[\delta k-n]-2\sum_i q_i-1}[c_{k,1}T_1^{4k+10}]^le^{-\frac{|\z-\w|^2}{4t}}}{[\delta k-n-2\sum_i q_i]\cdots[(l-1)(\delta k-n)-2\sum_i q_i]}\label{sec1-equa-6}.
\end{align}

We have the following estimate
\begin{align*}
&|D^\alpha_\z R_k^{l+1}(\z,\w,t)|\\
\le& \int^t_0 \int_{\C^n} D^\alpha_\z R_k(\z,\x,t-\tau)R_k^l(\x,\w,\tau)d\tau d\x\\
\le& \int^t_0 \int_{\C^n} |D^\alpha_\z R_k(\z,\x,t-\tau)|\sup_{\x\in \C^n} |R_k^l(\x,\w,\tau)|d\tau d\x\\
\le& e^{-\frac{|\z-\w|^2}{4t}}\int^t_0 (t-\tau)^{\delta k-n-1-\frac{3}{2}|\alpha|}[c^\alpha_{k,1}T^{4k+4|\alpha|+10}]\\
& \frac{t^{l[\delta k-n]-2\sum_i q_i-1}[c_{k,1}T^{4k+10}]^l }{[\delta k-n-2\sum_i q_i]\cdots[(l-1)(\delta k-n)-2\sum_i q_i]}d\tau\\
\le& \frac{e^{-\frac{|\z-\w|^2}{4t}}t^{(l+1)[\delta k-n]-2\sum_i q_i-1-\frac{3}{2}|\alpha|}[c^\alpha_{k,1}T^{4k+4|\alpha|+10}]
[c_{k,1}T^{4k+10}]^l}{[l(\delta k-n)-2\sum_i q_i-\frac{3}{2}|\alpha|]\cdots [\delta k-n-2\sum_i q_i] }\\
\end{align*}
In particular, this shows that the estimate (\ref{sec1-equa-6}) holds for all $l$ by induction argument.

Similarly, we have the estimate
\begin{align}
&|D^{l_0}_t R_k^{l+1}(\z,\w,t)|\nonumber\\
\le& \frac{e^{-\frac{|\z-\w|^2}{4t}}t^{(l+1)[\delta k-n]-2\sum_i q_i-1-3l_0}[c^{l_0}_{k,2}T^{4k+2l_0+10}]
[c_{k,2}T^{4k+10}]^l}{[l(\delta k-n)-2\sum_i q_i-3l_0]\cdots [\delta k-n-2\sum_i q_i] }
\end{align}

\begin{lm}Assume that $q_M-q_m<\frac{1}{3}$ and let $k\in \N$ satisfies $ k>\frac{3l_0+n+1+2\sum_i q_i}{\delta}$. Given $l_0\in \N$ and $|\alpha|\le 2l_0,\kappa\le l_0$. Denote by $a_k=\frac{2\sum_i q_i}{\delta k-n}, A=\frac{c_{k,1}T^{4k+10}}{\delta k-n}$. Then we have the estimate
\begin{align}
&|D^\alpha_\z R_k^{l+1}(\z,\w,t)|\le C\frac{t^{(l+1)[\delta k-n]-2\sum_i q_i-1-\frac{3}{2}|\alpha|}A^l }{\Gamma(l+1-a_k)}e^{-\frac{|\z-\w|^2}{4t}}\\
&|D^{\kappa}_t R_k^{l+1}(\z,\w,t)|\le C\frac{t^{(l+1)[\delta k-n]-2\sum_i q_i-1-3\kappa}A^l }{\Gamma(l+1-a_k)}e^{-\frac{|\z-\w|^2}{4t}},
\end{align}
where $C$ is a constant depending only on $k,l_0,T$ and $f$.
\end{lm}

\subsubsection*{\underline{Estimate of the parametrix}}

Note that $P_k=\E_0\E_1(\sum_{j=0}^k t^j U_j(\z,\w)$ and $P^{l+1}_k=P_k*R^{l+1}_k$. We have
\begin{align*}
&|D^\alpha_\z P_k|\\
\le & \sum_{j=0}^k \sum_{\alpha_1+\alpha_2+\alpha_3=\alpha}\frac{\alpha}{\alpha_1!\alpha_2!\alpha_3!} |D^{\alpha_1}\E_0 D^{\alpha_2}\E_1 D^{\alpha_3}(t^j U_j)|\\
&\le \sum_{j=0}^k \sum_{\alpha_1+\alpha_2+\alpha_3=\alpha}\frac{\alpha}{\alpha_1!\alpha_2!\alpha_3!} t^{-|\alpha_1|}|P^{|\alpha_1|}_0(\z-\w,t)|_+\cdot |G^{\alpha_2}(\z,\w,t)|\cdot t^j  F_j^{\alpha_3}(\z,\w,t)\E_0\E_1\\
&\le \sum_{j=0}^k \sum_{\alpha_1+\alpha_2+\alpha_3=\alpha}\frac{\alpha}{\alpha_1!\alpha_2!\alpha_3!} t^{-(1+\delta_M q_M)|\alpha_1|+(-\frac{1}{2}+\delta_3)|\alpha_2|+\delta_0|\alpha_3|+\delta j} T^{|\alpha_1|+3|\alpha_2|+|\alpha_3|+4j}\\
&\cdot |P^{|\alpha_1|}_0|_+(\z_t-\w_t,1)\cdot F^{\alpha_2}_g(\z_t,\w_t,1)\cdot F_j^{\alpha_3}(\z_t,\w_t,1)\E_1(\z_t,\w_t,1)\E_0(\z,\w,t).\\
&\le C_{\alpha,k} t^{-\frac{3}{2}|\alpha|+\delta}T^{4k+4|\alpha|}\E_1(\z_t,\w_t,1)\E_0(\z,\w,t)\\
&\le C_{\alpha,k}t^{-\frac{3}{2}|\alpha|+\delta}T^{4k+4|\alpha|}\E_0(\z,\w,t)
\end{align*}
Similarly, for $l_0\in \N$, we have
$$
|D^{l_0}_t P_k|\le C_{\alpha,k}t^{-l_0+\delta}T^{4k}\E_1(\z_t,\w_t,1)\E_0(\z,\w,t).
$$
Now we have
\begin{align*}
&|D^\alpha_\z P^{l+1}_k|\\
\le &\int^t_0 \int_{\C^n} |D^\alpha_\z P_k(\z,\x,t-\tau)||R^{l+1}_k(\x,\w,\tau)|d\x d\tau\\
\le &C_{\alpha,k} \int^t_0(t-\tau)^{-\frac{3}{2}|\alpha|+\delta}\tau^{(l+1)(\delta k-n)-2\sum_i q_i-1}d\tau \frac{A^{l} T^{4k+4|\alpha|}}{\Gamma(l+1-a_k)}\int_{\C^n}\E_0(\z,\x,t-\tau)\E_0(\x,\w,\tau)d\x\\
\le &C\E_0(\z,\w,t) \frac{A^{l} T^{4k+4|\alpha|}}{\Gamma(l+1-a_k)}\frac{t^{(l+1)(\delta k-n)+\delta-\frac{3}{2}|\alpha|-2\sum_i q_i }}{(l+1)(\delta k-n)-2\sum_i q_i}\\
\le &C\frac{A^{l+1}}{\Gamma(l+2-a_k)}t^{(l+1)(\delta k-n)+\delta-n-\frac{3}{2}|\alpha|-2\sum_i q_i }.
\end{align*}
Similarly, we obtain
$$
|D^\kappa_t P^{l+1}_k|\le C\frac{A^{l+1}}{\Gamma(l+2-a_k)}t^{(l+1)(\delta k-n)+\delta-n-\kappa-2\sum_i q_i }.
$$

Based on the estimate of $P^{l}_k$ for $l\in \N$, we get the following convergence result. 

\begin{thm}\label{sect-3-main-theorem-1} Let $f$ be a non-degenerate quasi-homogeneous polynomial on $\C^n$ satisfying $q_M-q_m<\frac{1}{3}$ and let $k\in \N$ satisfies $k>\frac{3l_0+n+1+2\sum_i q_i}{\delta}$. Fix $T>0$. Then for any $t\in (0,T]$, the series
$$
P(\z,\w,t)=\sum^\infty_{i=0}(-1)^i P^i_k(\z,\w,t)
$$
converges for any $(\z,\w)\in \C^n\times \C^n$, and further $P(\z,\w,t)$ has up to $2l_0$-order $\z$-derivatives and up to $l_0$ order $t$ derivatives. $P(\z,\w,t)$ is the unique heat kernel of the operator $\pat_t+\Delta_f$.  
\end{thm}

\section{Trace class and a vanishing theorem}

In this section, we return to the consideration of more general section-bundle system $(M,g,f)$ which was introduced in \cite{Fa}. We will show that under mild conditions that the heat kernel $e^{-t\Delta_f}$ is of trace class. Hence we can define the $i$-th zeta function $\Theta^i_f(s)$ for $\re(s)>2n,n=\dim_\C M$. Furthermore, if $f$ is a non-degenerate quasi-homogeneous polynomial defined on the Euclidean space $\C^n$, we can show that the first zeta function of $f$ vanishes. 

Denote by
$$
\Delta_f=\Delta_f^0+H^1_f,
$$
where
\begin{equation}
\Delta^f_0:=-\sum_{\mu \nu}g^{\nub \mu}\nabla_\mu \nabla_\nub +|\nabla
f|^2,
\end{equation}
and
\begin{equation}
H^1_f:=R+L_f,
\end{equation}
where $R$ is the curvature operator and $L_f $ is defined as in Section 2. 
 
\begin{lm}\label{lm:trac-equi} If $(M,g,f)$ is strongly tame, then $e^{-t\Delta_f}$ is of trace class if and only if $e^{-t\Delta_f^0}$ is of trace class.
\end{lm}

\begin{proof} By Theorem \ref{thm:intro-1}, $\Delta_f$ has purely discrete spectrum. There exists a complete orthonormal basis of $k$-forms $\{\phi_n\}_{n=1}^\infty$  in the domain $D(\Delta_f)$ such that $\Delta_n \phi_n=\lambda_n \phi_n$ with $\lambda_1\le \lambda_2\le \cdots $ and $\lambda_n\to \infty$.

$\lambda_n$ is given by the min-max principle, i.e.,
\begin{equation}
\lambda_n(\Delta_f)=\sup_{\varphi_1,\cdots,\varphi_{n-1}}U(\varphi_1,\cdots,\varphi_{n-1}),
\end{equation}
where
\begin{equation}
U(\varphi_1,\cdots,\varphi_m)=\inf_{\stackrel{\psi\in
D(\Delta_f);||\psi||=1}{\psi\in [\varphi_1,\cdots,\varphi_m]^\perp}}(\psi,
\Delta_f\psi).
\end{equation}
where $[\varphi_1,\cdots,\varphi_m]^\perp$ is the perpendicular complement
space of the subspace $[\varphi_1,\cdots,\varphi_m]$ generating by
$\varphi_1,\cdots,\varphi_m$.

On the other hand, by the inequality (95) of \cite{Fa}, one has 
\begin{equation}
|(H_f^1\varphi,\varphi)|\le
\epsilon(\Delta_f^0\varphi,\varphi)+(C_\epsilon+C_R)(\varphi,\varphi),
\end{equation}
which induces
\begin{equation}
(\Delta_f \varphi,\varphi)\ge (1-\epsilon)(\Delta_f^0\varphi,\varphi)-(C_\epsilon+C_R)(\varphi,\varphi),
\end{equation}
where $C_\epsilon,C_R$ are constants depending on $\epsilon $ and the geometry of $M$. 
By min-max principle, this gives
\begin{equation}
\lambda_n=\lambda_n(\Delta_f)\ge (1-\epsilon)\lambda^0_n-(C_\epsilon+C_R),\;\lambda^0_n:=\lambda_n(\Delta_f^0)
\end{equation}
Therefore, we have
\begin{equation}
\Tr_{L^2}e^{-t\Delta_f}=\sum_n e^{-\lambda_n}\le e^{C_\epsilon+C_R}\sum_n e^{-(1-\epsilon)t\lambda_n^0}=e^{C_\epsilon+C_R}\Tr_{L^2}(e^{-(1-\epsilon)t\Delta_f^0}).
\end{equation}
This shows that if $e^{-t\Delta_f^0}$ is of trace class, then $e^{-t\Delta_f}$ is of trace class. On the other hand, when acting on the 0-forms, the form Laplacian $\Delta_f$ is just $\Delta_f^0$, hence the vice versa conclusion also holds. 
\end{proof}

Now we consider the complex one dimensional section-bundle system $(\C, \frac{\sqrt{-1}}{2}dz\wedge d\bar{z},\frac{1}{2}z^2)$.  Set $\Delta_1=\Delta_{z^2/2}$. The heat kernel $K(z,w,t)$ of $\Delta_1^0$ is calculated explicitly in Section 2 and we have
\begin{equation}
\Tr e^{-t\Delta_1^0}=\Tr K(z,z,t)=\left(\frac{1}{2\sinh t/2}\right)^2.
\end{equation}
The above trace is also holomorphic w.r.t. $t\in \IH^+$. 
Similarly, w can consider the complex $n$-dimensional system $\left(\C^n, \frac{\sqrt{-1}}{2}\sum_{j=1}^n, \frac{1}{2}(z_1^2+\cdots+z_n^2)\right)$.
We denote its twisted Laplacian as $\Delta_{n}$, then we have 
\begin{equation}
\Tr e^{-t\Delta_{n}^0}=\left(\frac{1}{2\sinh t/2}\right)^{2n},
\end{equation}
which is obviously holomorphic w.r.t. $t\in \IH^+$.

\begin{lm} Assume that $f=f(z_1,\cdots,z_N)$ is a non-degenerate quasihomogeneous polynomial with each charge $q_i\le 1/2$. Then we have 
\begin{equation}
\Tr e^{-t\Delta_f^0}<\infty,\;\forall t\in \IH^+,
\end{equation}
and the trace is holomorphic w.r.t. $t\in \IH^+$. 
\end{lm}

\begin{proof}
The scalar Schr\"odinger equation is given by 
$$
\Delta_f^0=-\Delta^0+|\partial f|^2,
$$
where $|\partial f|^2=\sum_i |\frac{\partial f}{\partial z_i}|^2$.  By Corollary \ref{sect-2-lemm-2.9}, we have 
$$
\Delta_{C,n}^0-1\le \Delta_f^0,
$$
where $\Delta^0_{C,n}=\Delta^0_{(z_1^2+\cdots+z_N^2)/C}$.
By min-max principle of spectrum, this gives
$$
\lambda_k(\Delta_{C,n}^0)-1\le \lambda_k(\Delta_f^0).
$$
Hence we have
\begin{equation}
\Tr e^{-t\Delta_f^0}\le e \Tr e^{-t\Delta_{C,n}^0}<\infty,\;\forall t\in \IH^+,
\end{equation}
and obviously this trace is holomorphic w.r.t. $t\in \IH^+$. 
\end{proof}

\begin{crl} Let $f_s=f+\sum_{i=1}^m s_i g_i$ be the strong deformation of $f$, where $f$ is non-degenerate and quasihomogeneous, then we have 
\begin{equation}
\Tr e^{-t\Delta_{f_s}^0}<\infty,\;\forall t\in \IH^+,
\end{equation}
and the trace is holomorphic w.r.t. $t\in \IH^+$. 
\end{crl}
\begin{proof} This is due to the fact that every deformation polynomial can be upper bounded by a quantity depending on $|\nabla f|$. 
\end{proof}

\begin{ex}Consider the system $(\C^*, \frac{\sqrt{-1}}{2}\frac{dz}{z}\wedge {\frac{d\bar{z}}{\bar{z}}}, f=z+\frac{1}{z})$. If we change the coordinate by $z=e^w, w=w_1+iw_2\in \R\times S^1$, then the system is transformed to $(\R\times S^1, \frac{\sqrt{-1}}{2}dw\wedge d\bar{w}, f=e^w+e^{-w})$. The scalar twisted Laplacian is given by
\begin{equation}
\Delta_f^0=-\partial_w\partial_{\bar{w}}+4|\sinh w|^2=H_1+H_2.
\end{equation}
Here
\begin{align*}
&H_1=-\frac{1}{4}\partial^2_{w_1}+(e^{2w_1}+e^{-2w_1})\\
&H_2=-\frac{1}{4}\partial^2_{w_2}-2\cos 2w_2.
\end{align*}
Hence the eigenvalue of $\Delta_f^0$ has the form $\lambda_{1,n}+\lambda_{2,n}$, where $\lambda_{1,n}$ are the eigenvalues of $H_1$ and $\lambda_{2,n}$ are the eigenvalues of $H_2$. 

Note that $H_2$ is an second order elliptic operator defined on the compact manifold $S^1$, the number of its negative eigenvalues is finite and its heat kernel is of trace class, i.e., $\Tr e^{-t H_2}<\infty$. For the operator $H_1$, we have operator inequality:
$$
H_1\ge -\frac{1}{4}\partial^2_{w_1}+(2+4w_1^2):=H_0+2,
$$
which shows that $\lambda_{1,n}\ge \lambda_n(H_0)+2$. 

Since $H_0$ is just the harmonic oscillator, the heat kernel of $H_1$ is of trace class. Therefore, we have 
$$
\Tr e^{-t\Delta_f^0}\le C \Tr e^{-tH_1}\Tr e^{-tH_2}<\infty.
$$
The trace $\Tr e^{-t\Delta_f^0}$ is holomorphic w.r.t. $t\in H^+$, since it is controlled by the trace of the heat kernel of a harmonic oscillator 
and an second order elliptic operator on a compact manifold. 
\end{ex}

let $f\in \C[z_1,\cdots,z_n,z_1^{-1},\cdots,z_n^{-1}]$ be a Laurent
polynomial defined on $(\C^*)^n$ having the form
\begin{equation}
f(z_1,\cdots,z_n)=\sum_{\alpha=(\alpha_1,\cdots, \alpha_n)\in \Z^n}
c_\alpha z^\alpha.
\end{equation}

\begin{lm}\label{lm:trac-laur} Assume that the above $f$ is convenient and non-degenerate, Then we have
\begin{equation}
\Tr e^{-t\Delta_f^0}<\infty,
\end{equation}
and this trace is holomorphic w.r.t $t\in H^+$. 
\end{lm}

\begin{proof} Note that under the change of the coordinate $z_i=e^{w_i}$, the operator $\Delta_f^0$ is defined on $(\R^1\times S^1)^n$ and have the form
\begin{equation}
\Delta_f^0=-\sum_i \partial_{w_i}\partial_{\bar{w}_i}+\sum_i \left|\sum_\alpha a_\alpha \alpha_i e^{\langle \alpha,w\rangle}  \right|^2=-\sum_i \partial_{w_i}\partial_{\bar{w}_i}+F(w).
\end{equation}
In \cite[Proposition 2.47]{Fa}, we have proved that there exists a constant $C_0$ such that 
\begin{equation}
 F(w)\ge C_0 \left|e^{M(\beta)|w|} \right|^2-C.
\end{equation}
Here $M(\beta)|w|=\re(\langle w, \alpha_s\rangle)$ and $\alpha_s$ is one of the special vectors representing the vertices of the Newton polyhedron of $f$ such that the number $M(\beta)|w|$ is the largest among the numbers $\re(\langle w, \alpha\rangle)$ for any vertices $\alpha$. Since $0$ is in the interior of the Newton polyhedron $\Delta(f)$, there is a number $1\ge \theta_0>0$ such that for any $w$, 
\begin{equation}
F(w)\ge C_0 e^{2\theta_0|\re(w)|}-C\ge c_0|\re(w)|^2-C=c_0\sum_i |\re(w_i)|^2-C.
\end{equation}
Denote by
\begin{equation}
H^0=-\sum_i \partial_{w_i}\partial_{\bar{w}_i}+c_0\sum_i |\re(w_i)|^2-C
\end{equation}
Then by the fact that 
\begin{equation}
\lambda_n(\Delta_f^0)\ge \lambda(H^0),
\end{equation}
we have 
$$
\Tr e^{-t\Delta_f^0}\le C\Tr e^{-t H^0}<\infty,
$$
and $\Tr e^{-t\Delta_f^0}$ is holomorphic w.r.t $t\in H^+$. 

\end{proof}

Let $f$ be a convenient and nondegenerate Laurent polynomial. The
$\C$-vector space
$Q_f=\C[z_1,\cdots,z_n,z_1^{-1},\cdots,z_n^{-1}]/J_f$ is finite
dimensional and its dimension $\mu(f)$ is the sum of the Milnor
numbers of $f$ at each critical point. Let $\{g_i,i=1,\cdots,r\}$ be
the set of monomials such that their corresponding lattice points
are contained in the interior of the Newton polyhedron of $f$ and
injects to the Jacobi space $Q_f$. Then the deformation
\begin{equation}\label{eq:deform-Laurent-1}
f_t(z):=f(z)+\sum_{j=1}^r s_j g_j.
\end{equation}
is called the subdiagram deformation. It was shown in \cite[PP. 23]{Sab} that for any $s=(s_1,\cdots,s_r)\in \C^r$ the Laurent
polynomial $f_s(z)$ is convenient and nondegenerate. Therefore by
Lemma \ref{lm:trac-laur}, we have the following
result.

\begin{crl} Let $f$ be a convenient and nondegenerate Laurent polynomial defined on $(\C^*)^n$ and $f_s(z)$ is the subdiagram deformation of $f$, then for any $s$ we have 
\begin{equation}
\Tr e^{-t\Delta_{f_s}^0}<\infty,\forall t\in \IH^+,
\end{equation}
and the trace is holomorphic w.r.t. $t\in \IH^+$. 
\end{crl}

By Lemma \ref{lm:trac-equi} and the analysis of the above cases, we obtain the following conclusions:

\begin{thm}\label{sect-trac-main-1} Assume that $(M,g,f)$ is of strongly tame. Then $\Tr_{L^2}e^{-t \Delta_f}<\infty$ for any $t\in \IH^+:=\{t\in \C|\re(t)>0\}$ and it is holomorphic with respect to $t\in \IH^+$ if $f$ is one of the following cases:
\begin{enumerate}
\item $f=f(z_1,\cdots,z_n)$ is a non-degenerate quasi-homogeneous function;
\item $f=f_s(z_1,\cdots,z_N)$ is a strong deformation of $f$ in (1);
\item $f=f(z_1,\cdots,z_n)$ is a convenient and non-degenerate Laurent polynomial defined on $(\C^*)^n$;
\item $f=f_s(z_1,\cdots,z_n)$ is the subdiagram deformation of $f$ in (3).
\end{enumerate}
\end{thm}

Now we proved that $\Tr e^{-t\Delta_f}$ is bounded by the trace of harmonic oscillator. Since the first term in the expansion of the heat kernel $e^{\Delta_{n}^0}$ has the order $t^{-2n}$, we can define the zeta functions as follows.

\begin{df}Under the assumption about $f$ in Theorem \ref{sect-trac-main-1}, we can define the $i$-th Zeta function for $\re(s)> 2n 
$:
\begin{equation}
\Theta^i_f(s)=\frac{1}{2\Gamma(s)}\int^\infty_0 \str(N^i (e^{-t\Delta_f}-\Pi))t^{s-1}dt, i=1,2,\cdots,
\end{equation}
where $\Pi: L^2\Lambda^*(\C^n)\to \Hc^*$ is the projection operator into the space of $L^2$ harmonic forms and $N$ is the number operator. 
\end{df}

It is immediate that $\Theta^i_f(s)$ is analytic in the domain $\{\re(s)\ge 2n\}$. At first, we have a vanishing theorem.

\begin{thm}\label{sec4:thm-1} Let $f$ be a non-degenerate quasi-homogeneous polynomial, then
\begin{equation}
\Theta^1_f(s)\equiv 0
\end{equation}
\end{thm}

\begin{proof}
We can rewrite $\Theta^1_{f}(s)$ as 
$$
\Theta^1_{f}(s)=\frac{1}{2}\sum_{k=0}^{2n}(-1)^k k\frac{1}{\Gamma(s)}\int^\infty_0 \Tr(e^{-t\Delta^k_{f}}-\Pi)t^{s-1}dt.
$$

Let $\D^k$ be the space of all smooth $k$-forms which tend exponentially to zero at the infinity. Define $\D^k_1=\bpat_f \D^{k-1}$ and $\D^k_2=\bpat_f^\dag \D^{k+1}$. Then by Hodge theorem, $\D_k=\Hc^k\oplus \D_1^k\oplus \D_2^k $. Note that $\D^0_1=0,\D^{2n}_2=0$. 

On $\D_1^k$, the Laplacian $\Delta^k_f=\bpat_f\bpat^\dag_f$ and is invariant under the action of heat kernel $e^{-t\Delta^k_f}$, whose restriction  is denoted by $(e^{-t\bpat_f\bpat^\dag_f})_{1,k}$. Similarly, we can define $(e^{-t\bpat^\dag_f\bpat_ f})_{2,k}$ to be the restriction of the heat kernel to $\D^k_2$.  We have
$$
\Tr(e^{-t\Delta^k_f}-\Pi^k)=\Tr(e^{-t\bpat_f\bpat^\dag_f})_{1,k}+\Tr(e^{-t\bpat^\dag_f\bpat_f})_{2,k}
$$
Now $\bpat_f$ is $1-1$ from $\D_2^{k}$ onto $\D_1^{k+1}$. In fact, we have 
$$
\bpat_f(e^{-t\bpat^\dag_f\bpat_f})_{2,k}=(e^{-t\bpat_f\bpat^\dag_f})_{1,k+1}\bpat_f.
$$
Therefore, we have
\begin{equation}\label{sec4:eq-1}
\Tr(e^{-t\bpat^\dag_f\bpat_f})_{2,k}=\Tr(e^{-t\bpat_f\bpat^\dag_f})_{1,k+1}.
\end{equation}
Since $\Delta_f=\bpat_f\bpat^\dag_f+\bpat^\dag_f\bpat_f=\pat_f\pat^\dag_f+\pat^\dag_f\pat_f$, the analogous identity (\ref{sec4:eq-1}) also holds for $\pat_f,\pat^\dag_f$:
\begin{equation}
\Tr(e^{-t\pat^\dag_f\pat_f})_{2,k}=\Tr(e^{-t\pat_f\pat^\dag_f})_{1,k+1}.
\end{equation}

On the other hand, $*$-operator provides another symmetry:
$$
*\bpat_f\bpat^\dag_f=\pat^\dag_{-f}\pat_{-f}*,\;*\bpat^\dag_f\bpat_f=\pat_{-f}\pat^\dag_{-f}*,
$$
therefore we have 
\begin{align*}
&\Tr(e^{-t\bpat_f\bpat^\dag_f})_{1,k}=\Tr(e^{-t\pat^\dag_{-f}\pat_{-f}})_{2,2n-k},\\
&\Tr(e^{-t\bpat^\dag_f\bpat_f})_{2,k}=\Tr(e^{-t\pat_{-f}\pat^\dag_{-f}})_{1,2n-k}.
\end{align*}
We have 
\begin{align*}
Q_f:=&\sum_{k=0}^{2n}(-1)^k k\Tr(e^{-t\Delta^k_f}-P^k)\\
=&\sum_{k=0}^{2n}(-1)^k k\left( \Tr(e^{-t\bpat_f\bpat^\dag_f})_{1,k}+\Tr(e^{-t\bpat^\dag_f\bpat_f})_{2,k}\right)\\
=&\sum_{k=0}^{2n}(-1)^k k\left( \Tr(e^{-t\pat^\dag_{-f}\pat_{-f}})_{2,2n-k}+\Tr(e^{-t\pat_{-f}\pat^\dag_{-f}})_{1,2n-k}\right)\\
=&\sum_{k=0}^{2n}(-1)^{(2n-k)} (2n-k)\left( \Tr(e^{-t\pat^\dag_{-f}\pat_{-f}})_{2,k}+\Tr(e^{-t\pat_{-f}\pat^\dag_{-f}})_{1,k}\right)\\
=&2n\sum_{k=0}^{2n}(-1)^k\left( \Tr(e^{-t\pat^\dag_{-f}\pat_{-f}})_{2,k}+\Tr(e^{-t\pat_{-f}\pat^\dag_{-f}})_{1,k}\right)-Q_{-f}.
\end{align*}
Hence, we have
\begin{align*}
Q_f+Q_{-f}=&2n\sum_{k=0}^{2n}(-1)^k\left( \Tr(e^{-t\pat^\dag_{-f}\pat_{-f}})_{2,k}+\Tr(e^{-t\pat_{-f}\pat^\dag_{-f}})_{1,k}\right)\\
=&2n\sum_{k=0}^{2n}(-1)^k\left( \Tr(e^{-t\pat_{-f}\pat^\dag_{-f}})_{1,k+1}+\Tr(e^{-t\pat_{-f}\pat^\dag_{-f}})_{1,k}\right)\\
=&0
\end{align*}
So
$$
\Theta^1_f+\Theta^1_{-f}=0.
$$
Since $f$ is quasi-homogeneous, there exist integers $d,k_1,\cdots, k_n$ such that for any $\lambda\in \C^*$
$$
\lambda^d f(z_1,\cdots,z_n)=f(\lambda^{k_1}z_1,\cdots,\lambda^{k_n}z_n).
$$
Hence we can take a $\xi$ satisfying $\xi^d=-1$. Define the matrix $I_\xi=\diag(\xi^{k_1},\cdots,\xi^{k_n})$ which gives an isomorphism $I_\xi:\C^n\to \C^n$ such that  $f(I_\xi\cdot \z)=-f(\z)$. So we have the pull back map $I_\xi^*:\Omega^\bullet(\C^n)\to \Omega^\bullet(\C^n)$. It is easy to check that 
\begin{equation}
(I_\xi^*)^{-1}\circ\Delta_{-f}\circ I_\xi^*=\Delta_f.
\end{equation} 
Hence if $p(\z,\w,t)$ is the heat kernel of $\Delta_f$, then $I_\xi^*P(\z,\w,t)=P(I_\xi(\z),I_\xi(\w),t)$ is the heat kernel of $\Delta_{-f}$. This implies that $\Theta^1_{-f}=\Theta^1_f$ and then
$$
\Theta^1_{f}(s)=\Theta^1_{-f}(s)=0
$$. 

\end{proof}

\section{Index Theorem and torsion invariants}

In this section, we always assume that $f$ is a non-degenerate quasi-homogeneous polynomial. Above all, we prove a corresponding index theorem for the associated chain complex. We follow loosely the procedure to prove a local index theorem for the Dirac operator on a closed compact manifold (see for example~\cite{BGV}). However the details are vastly different. By analyzing the asymptotic expansion of the heat kernel as $t\to 0$, we can extend the zeta functions $\Theta^i_f(s)$ (for $i\ge 2$) meromorphically to $\C$ such that they are regular at $s=0$. Finally, we define the $i$-th torsion invariants $T^i_f$ of singularities via this analytic extension at $s=0$.  As an example, we compute the torsion invariants of $A_r$-singularities. When $r=1$, our torsion invariants is related to the famous Riemann zeta function. Note that if we use the expansion of heat kernel we always assume the condition $q_M-q_m<1/3$ to grantee the conclusion of Theorem \ref{sect-3-main-theorem-1}. 

First we prove a McKean-Singer formula as follows:
\begin{prop}
$${d\over  dt}{\str}(\exp (-t\Delta_{f})=0$$ 
\end{prop}
\begin{proof}
Notice that $\Delta_{f}=D_{f}^{2}$, where $D_f=\bpat_f+\bpat_f^\dag$. We get
$${d\over  dt}{ \str}(\exp (-t\Delta_{f})={\rm Str}(D_{f}^{2}\exp (-t\Delta_{f})=\frac{1}{2} {\rm Str}[D,D\exp (-t\Delta_{f}]=0.$$ 
\end{proof}

\begin{lm}
$$\lim_{t\to\infty}{\str}(\exp (-t\Delta_{f})=\e (H^{*}(\Omega^{*},D_{f}),$$where the right hand side is defined as the euler characteristic of the cohomology associated to the chain complex $(\Omega^{*},D_{f})$, which equals to $(-1)^{n}\mu$.
\end{lm}

In the rest of the section, we analyze the quantity $\lim_{t\to 0}{\rm Str}(\exp (-t\Delta_{f})$. By Section 4, we have an asymptotic expansion of the corresponding heat kernel when $t$ is small, which a piori does not warrant the existence of the above-mentioned limit.

\begin{thm}\label{index}
$$
\lim_{t\to 0}{ \str}(\exp (-t\Delta_{f})=\frac{(-1)^{n}}{\pi^n}\int_{\C^{n}}\exp (-|\pat f|^{2})|\det \partial^{2}f|^{2}d\text{vol}.
$$
\end{thm}

\begin{proof} We have the expansion of the trace of the heat kernel of $m$ forms as follows:
$$
P(\z,\w,t)=\E(\z,\w,t)\E_1(\z,\w,t)\sum_{j=0}^\infty t^j U_j(\z,\w)\in \Omega^m(\C^n_\z\times \C^n_\w,[0,T]),
$$
where 
\begin{align*}
&U_0\equiv I\\
&U_1(\z,\w)=-\int^1_0 B(s(\z-\w)+\w)ds,
\end{align*}
and $U_j$ is inductively determined by 
\begin{equation}\label{sec4:eq-1}
U_j(\z,\w)=\frac{1}{j}\{\Delta_\z U_{j-1}-B(\z)\cdot U_{j-1}(\z,\w)-\Delta_\z g\cdot U_{j-2}-2\nabla_\z g\cdot \nabla_\z U_{j-2}+(\nabla g)^2U_{j-3} \}.
\end{equation}
Set $\z=\w$, we obtain
$$
P(\z,\z,t)=(4\pi t)^{-n}e^{-t|\pat f|^2}\sum_{i=0}^\infty t^j U_j(\z,\z).
$$

Note that 
$$
B=L_f=\frac{1}{4}(f_{\mu l}+\overline{f_{\mu l}})(\hat{\bar{c}}_\mu c_l-\bar{c}_\mu \hat{c}_l)+\frac{1}{4}(f_{\mu l}-\overline{f_{\mu l}})(\hat{\bar{c}}_\mu \hat{c}_l-\bar{c}_\mu c_l).
$$
According to Lemma \ref{sec2:lm-0}, $\str B^k\neq 0$ only if $k=2nj $ for some $j\in \N$. An induction argument shows that $\str U_j\equiv 0$ for $1\le i\le 2n-1$, and $\str(U_{2n})=\str(L_f^{2n})/(2n)!$ due to Proposition \ref{sec1:prop-1}. For $j\ge 1$, by applying induction to (\ref{sec4:eq-1}), we can prove that 
$$
t^{2n+j}\Str (e^{-t|\pat f|^2} U_{2n+j}(\z,\z))\le O(t^{n+\delta}), \delta=\frac{1-3(q_M-q_m)}{2(1-q_M)}.
$$
hence 
\begin{align*}
&\lim_{t\to 0}\Str(e^{-t\Delta_f})=\lim_{t\to 0}\Str(P(\z,\z,t))\\
=&\lim_{t\to 0}(4\pi t)^{-n}\int_{\C^n}e^{-t|\pat f|^2}\str\left(\frac{ t^{2n}  L_f^{2n}}{(2n)!}\right)d\text{vol}\\
=&\lim_{t\to 0}\frac{(-1)^n}{\pi^n}\int_{\C^n}e^{-t|\pat f|^2} t^n |\det \pat^2 f|^2 d\text{vol}\\
=&\frac{(-1)^n}{\pi^n}\int_{\C^n}e^{-|\pat f|^2} |\det \pat^2 f|^2 d\text{vol}.
\end{align*}
\end{proof}

\begin{crl}\label{sect-5-coro-index1} Under the conditions of Theorem \ref{sect-3-main-theorem-1}, The Milnor number $\mu(f)$ is given by the Gaussian type integral
\begin{equation}
\mu(f)=\frac{t^n}{\pi^n}\int_{\C^{n}}\exp (-t|\pat f|^{2})|\det \partial^{2}f|^{2}d\text{vol},\forall t>0,
\end{equation}
if we take $f(z)=z^2/2$, then we obtain the well-known formula:
\begin{equation}
\pi^n=\int_{\C^n}e^{-|\z|^2}d\text{vol}.
\end{equation}
\end{crl}

\begin{lm}\label{sect-6-lemm-gaus-1}
Let $\alpha,\beta$ be multiple indices. Then the following integration has the property:
\begin{equation}
\int_{\C^n} \exp (-t|\pat f|^{2}) \z^\alpha\bar{\z}^\beta d\text{vol}\neq 0\;\text{iff}\;\alpha=\beta.
\end{equation}
\end{lm}

Now we begin the definition of the torsion invariants. 

At first, we have the expansion formula for the $k$-form heat kernel:
\begin{equation}
P^{(k)}(\z,\w,t)=\E_0\E_1(\sum_{i=0}^\infty t^i U^{(k)}_i(\z,\w))=\frac{1}{(4\pi t)^n}e^{-\frac{|\z-\w|^2}{4}}e^{-t g(\z,\w)}(\sum_{k=0}^\infty t^i U^{(k)}_i(\z,\w)).
\end{equation}

\begin{lm}\label{sect-6-lemm-expa-1}
\begin{align}
&\Tr(e^{-t\Delta^k_f})=\int_{\C^n}\tr (P^{(k)}(\z,\z,t)) d\z\nonumber\\
=&\frac{1}{(4\pi t)^n} \int_{\C^n} e^{-t |\pat f|^2(\z)}\left(\sum_{k=0}^\infty t^i \tr(U^{(k)}_i(\z,\z))\right)d\z=\sum_{i=1}^\infty a^k_i t^{\alpha^k_i},
\end{align}
where $\alpha^k_i\in \Q$ satisfying:
\begin{enumerate}
\item $\alpha^k_1=-(n+2|\q|)<\alpha^k_2<\cdots<\alpha^k_{i_{k0}}=0<\alpha^k_{i_{k0}+1}<\cdots$.
\item $\alpha^k_i$ has the form $m+\alpha\cdot q-2|\q|-n, m\in \N$ and  $\alpha>0$ is a multiple index. 
\end{enumerate}
\end{lm}

\begin{proof}It was shown by Theorem \ref{sect-3-main-theorem-1} that the trace of the heat kernel has convergent series w. r. t. $t$. At first, we have $\alpha^k_1\ge -2n$ by comparising with $n$-dimensional harmonic oscillator in Section 5. It suffices to prove that the power indices $\alpha^k_i\in \Q$ and $\alpha^k_1=-(n+2|\q|)$. By Lemma \ref{sect-6-lemm-gaus-1}, we only need to consider the expansion of the integration:
$$
J(t)=\int_{\C^n} \exp (-t|\pat f|^{2}) \z^\alpha\bar{\z}^\alpha d\text{vol}_{\z}.
$$
We have
\begin{align*}
J(t)&=\int_{\C^n}e^{-\sum_i t^{(1-q_i)}|\pat_i f|^2} e^{\sum_i (1-t^{q_i})t^{1-q_i}|\pat_i f|^2}\z^\alpha\bar{\z}^\alpha d\text{vol}_{\z}\\
&=\int_{\C^n}e^{-|\pat f |^2(\z_t)} \sum_{k=0}^\infty \frac{1}{k!} \left( \sum_i (1-t^{q_i}) |\pat_i f|^2(\z_t) \right)^k \z_t^\alpha\bar{\z}_t^\alpha d\text{vol}_{\z_t} t^{-(2\alpha\cdot \q+2|\q|)},
\end{align*}
where $\z_t=(z_1^t,\cdots,z^t_n),z_i^t=t^{q_i/2}$. Now it is easy to see that the conclusion hold. 
\end{proof}

The following conclusion is obvious. 
\begin{lm}
$\int^\infty_\frac{1}{2} \tr(e^{-t\Delta^k_f}-\Pi)t^{s-1}dt$ is a holomorphic function of $s\in \C$. 
\end{lm}

On the other hand, for any $0<\epsilon<\frac{1}{2}$ we have the computation:
\begin{align*}
&\int^\frac{1}{2}_\epsilon \tr(e^{-t\Delta_f^k}-\Pi)t^{s-1}dt=\sum^{\infty}_{i\neq i_0}\frac{a^k_i}{\alpha^k_i+s}((\frac{1}{2})^{\alpha^k_i+s}-\epsilon^{\alpha^k_i+s}).
\end{align*}
For any $1\le i\le i_{k0}-1$ and $\re(s)<-\alpha^k_{i_{k0}-1}$,  the term $\epsilon^{\alpha_i+s}\to \infty$ as $\epsilon\to 0$. This is called the ultraviolet disvergence phenomenon. The zeta function $\Theta_f^i(s)$ is analytical in the domain $\re(s)>-\min_{k}(\alpha^k_1)$. To obtain the analytical continuation, instead we consider the renormalized zeta function of scale $[\epsilon,L]$: 
\begin{equation}
\Theta^{i}_{f}(s;\epsilon,L):=\frac{1}{2\Gamma(s)}\left(\int^L_\epsilon \str(N^i (e^{-t\Delta_{\tau,\bu}}-\Pi))t^{s-1}dt+\sum_{k=0}^{2n}(-1)^k k^i\sum^{i_{k0}-1}_{i=1}\frac{a^k_i}{\alpha^k_i+s}\epsilon^{\alpha^k_i+s}\right).
\end{equation}
We denote the first term by $\Theta^i_{f,1}(s;\epsilon,L)$ and the second term by $\Theta^i_{f,c}(s;\epsilon,L)$, i.e., the cancellation factor of $\Theta^i_{f,1}(s;\epsilon,L)$. 

The renormalized zeta function is defined as 
\begin{equation}
\Theta^{R,i}_{f}(s):=\lim_{L\to \infty}\lim_{\epsilon\to 0}\Theta^{i}_{f}(s;\epsilon,L).
\end{equation}
Now the function $\Theta^{R,i}_{f}(s)$ is a meromorphic function on $\C$ with simple pole $\alpha^k_i$ for $1\le i< i_{k0}$. Therefore $\Theta^{R,i}_{f}(s)$ is analytical at $s=0$, we can define the torsion invariants.

\begin{df}For $i\in\N$, the $i$-th torsion invariants $T^i(f)$ of the singularity $f$ is defined as 
\begin{equation}
\log T^i(f)=-(\Theta^{R,i}_f)'(0).
\end{equation}
\end{df}

\begin{prop}Let $(\C^{n_1},f_1(\z))$ and $(\C^{n_2},f_2(\w))$ be two non-degenerate quasihomogeneous polynomial which satisfying the condition $q_M-q_m<1/3$ for both $\z$ and $\w$, then we have the sum of the singularity $(\C^{n_1+n_2}, f_1(\z)+f_2(\w))$ and the identity of torsions
\begin{equation}
\log T^2(f_1\oplus f_2)=(-1)^{n_1}\mu(f_1)\log T^2(f_2)+(-1)^{n_2}\mu(f_2)\log T^2(f_1)  
\end{equation}
\end{prop}

\begin{proof} We want to prove
\begin{equation}\label{sect6-zeta-func-1}
\Theta^2_{f_1\oplus f_2}(s)=(-1)^{n_1}\mu(f_1)\Theta^2_{f_2}(s)+(-1)^{n_2}\mu(f_2)\Theta^2_{f_1}(s),
\end{equation}
and the conclusion of this proposition follows from this fact. 

Let $\Pi^{p}=\sum_{p_1+p_2=p}\Pi^{p_1}_1\otimes \Pi^{p_2}_2$ be the projection from $L^2\Lambda^p(\C^{n_1}\times \C^{n_2})$ to the space $\Hc^p$ of $p$-harmonic forms on $\C^{n_1}\times \C^{n_2}$. We have
\begin{align*}
&\sum_{p_1,p_2}(-1)^{p_1+p_2}(p_1+p_2)^2 \Tr(e^{-t \Delta^{p_1}_{f_1}\otimes \Delta^{p_2}_{f_2}}-\Pi^{p_1}\otimes \Pi^{p_2})\\
=&\sum_{p_1,p_2}(-1)^{p_1+p_2}[p_1^2+2p_1p_2+p_2^2]\Tr\left( (e^{-t\Delta^{p_1}_{f_1}}-\Pi^{p_1}_1)\otimes e^{-t\Delta_2^{p_2}}+\Pi_1^{p_1}\otimes (e^{-t\Delta^{p_2}_2}-\Pi^{p_2}_2 )\right)\\
=&\sum_{p_1,p_2} (-1)^{p_1+p_2}p_1^2\left\{\Tr (e^{-t\Delta_{f_1}^{p_1}}-\Pi_1^{p_1})\Tr e^{-t\Delta_{f_2}^{p_2}} +\Tr \Pi_1^{p_1} \Tr (e^{-t\Delta^{p_2}_{f_2}}-\Pi^{p_2}_2 )\right\}\\
+&2\sum_{p_1,p_2}(-1)^{p_1+p_2}p_1p_2\left\{\Tr (e^{-t\Delta_{f_1}^{p_1}}-\Pi_1^{p_1})\Tr e^{-t\Delta^{p_2}_{f_2}}+\Tr \Pi_1^{p_1} \Tr (e^{-t\Delta^{p_2}_{f_2}}-\Pi^{p_2}_2 )  \right\}\\
+&\sum_{p_1,p_2} (-1)^{p_1+p_2}p_2^2\left\{\Tr e^{-t\Delta_{f_1}^{p_1}}\Tr (e^{-t\Delta^{p_2}_{f_2}}-\Pi^{p_2}_2)+\Tr (e^{-t\Delta^{p_1}_{f_1}}-\Pi^{p_1}_1 )\Tr \Pi^{p_2}_2\right\}\\
=&\left( \sum_{p_1} (-1)^{p_1} p_1^2 \Tr (e^{-t\Delta_{f_1}^{p_1}}-\Pi_1^{p_1})\right) \left( \sum_{p_2}(-1)^{p_2}\Tr e^{-t\Delta_{f_2}^{p_2}}\right) \\
+&\left( \sum_{p_1}(-1)^{p_1}\Tr e^{-t\Delta_{f_1}^{p_1}}\right) \left( \sum_{p_2} (-1)^{p_2} p_2^2 \Tr (e^{-t\Delta_{f_2}^{p_2}}-\Pi_2^{p_2}) \right). 
\end{align*}
Here we used index formula and the result in Theorem \ref{sec4:thm-1}. Now this identity gives (\ref{sect6-zeta-func-1}). 

\end{proof}

\subsubsection*{\underline{A special case: $A_r$ singularity}}

Consider a holomorphic function $f$ on the complex 1-dimensional Euclidean space $(\C,g=\frac{i}{2}dz\otimes d\zb)$. Then 
$$ 
\Delta_f=\Delta_{\bpat}-2(f''\iota_{\pat_\zb}dz\wedge+\overline{f''}\iota_{\pat_z}d\zb\wedge)+2|f'|^2.
$$
Hence when acting on $0$ or $2$ forms, $\Delta_f$ is the following scalar operator:
$$
\Delta_f^0=\Delta^2_f=\Delta_{\bpat}+2|f'|^2=2(-\pat_z\pat_\zb+|f'|^2).
$$
By \cite[Theorem 2.43]{Fa}, if $f$ is strongly tame, then $\Delta^0_f$ has purely discrete spectrum $\{\lambda_i,i=1,2,\cdots\}$. Hence we have 
\begin{align*}
 \Theta^m_f(s)&=\frac{1}{2\Gamma(s)}\int^\infty_0 \left(\sum_{k=0}^{2n}(-1)^k k^m\Tr(e^{-t\Delta^k_f}-\Pi)\right)t^{s-1}dt\\
 &=\frac{2^m-2}{2\Gamma(s)}\int^\infty_0 \Tr(e^{-t\Delta^0_f})t^{s-1}dt\\
 &=(2^{m-1}-1)\sum_{i=1}^\infty \lambda_i^{-s}
\end{align*}
In particular, we have 
\begin{equation}
\Theta^2_f(s)=\sum_{i=1}^\infty \lambda_i^{-s}
\end{equation}

Let $f=\frac{\tau z^2}{2}$ (i.e., complex 1 dimensional harmonic oscillator case), we have 
\begin{align*}
\Delta_f^0=&2(-\pat_z\pat_\zb+|\tau|^2|z|^2)\\
=&\frac{1}{2}\left\{(-\pat^2_x+4|\tau|^2x^2)+(-\pat_y^2+4|\tau|^2y^2)\right\}.
\end{align*}
Since $-\pat^2_x+4|\tau|^2x^2$ has eigenvalues $2|\tau|(1+2k),k=0,1\cdots$, the eigenvalue of $\Delta_f^0$ is given by 
\begin{equation}
\lambda_{k,l}=2|\tau|(1+k+l),\;k,l=0,1,\cdots,
\end{equation}
which implies that 
\begin{equation}
\Theta^2(s)=(2|\tau|)^{-s}\zeta(s-1),
\end{equation}
where $\zeta(s)$ is the {\em Riemann Zeta function}. Hence we have
\begin{equation}
 T^2(f)=(2|\tau|)^{-\frac{1}{12}}e^{-\zeta'(-1)}.
\end{equation}

\providecommand{\bysame}{\leavevmode\hbox
to3em{\hrulefill}\thinspace}

\end{document}